\documentclass[10pt,letter,final]{IEEEtran}
\usepackage{amsthm,amssymb,graphicx,multirow,amsmath,color,amsfonts}
\usepackage[update,prepend]{epstopdf}
\usepackage[noadjust]{cite}
\usepackage{tikz}
\usetikzlibrary{arrows,calc}		
\usepackage{bbm} 
\usepackage{pdfpages}
\usepackage{tabulary}
\usepackage{multirow}
\usepackage{comment}
\usepackage{dsfont}
\usepackage{pgf,tikz}
\usepackage{paralist}
\usepackage{amsfonts}
\usepackage{setspace,lipsum}
\usepackage{url}
\usepackage{soul}
\usepackage{bigints}
\usepackage{mathtools}
\usepackage[english]{babel}
\usepackage[autostyle]{csquotes}
\usepackage{textcomp}
\usepackage{varwidth}
\usepackage{tcolorbox}
\usepackage{lipsum}
\usepackage[normalem]{ulem}

\let\OLDthebibliography\thebibliography
\renewcommand\thebibliography[1]{
  \OLDthebibliography{#1}
  \setlength{\parskip}{0pt}
  \setlength{\itemsep}{0pt plus 0.3ex}
}

\newcommand{\pgfl}{\mathtt{PGFL}}
\newcommand{\cdf}{\mathtt{CDF}}

\newcommand{\pdf}{\mathtt{pdf}}
\newcommand{\pmf}{\mathtt{pmf}}

\newcommand\numberthis{\addtocounter{equation}{1}\tag{\theequation}}

\newcommand{\x}{\mathbf{x}}
\newcommand{\y}{\mathbf{y}}
\newcommand{\z}{\mathbf{z}}

\usepackage{amssymb,amsmath,amsfonts,amsthm}
\usepackage{mathrsfs}
\usepackage{cite}
\usepackage{array}
\usepackage{tabularx}
\usepackage{color}
\usepackage{mathtools}
\usepackage{subfigure}
\usepackage{mathtools}
\usepackage{tikz,pgf}
\usetikzlibrary{calc,positioning,mindmap,trees,decorations.pathreplacing}
\usepackage{graphicx}
\usepackage{bbm}
\usepackage[utf8]{inputenc}
\usepackage{algorithm}
\DeclareMathOperator{\sinc}{sinc}
\renewcommand{\IEEEQED}{\IEEEQEDopen}

\def\nb0{{\mathbf{0}}}
\def\nb1{{\mathbf{1}}}

\def\ncalA{{\mathcal{A}}}
\def\ncalB{{\mathcal{B}}}

\def\ncalD{{\mathcal{D}}}
\def\ncalE{{\mathcal{E}}}

\def\ncalJ{{\mathcal{J}}}

\def\ncalV{{\mathcal{V}}}

\def\sinc{{\rm sinc}}

\newtheorem{lemma}{Lemma}
\newtheorem{thm}{Theorem}

\newtheorem{theorem}{Theorem}

\newtheorem{cor}{Corollary}

\newtheorem{remark}{Remark}
\newtheorem{assumption}{Assumption}

\def\E{\mathbb{E}}

\def\P{\mathbb{P}}

\def\p{p}

\def\R{\mathbb{R}}

\def\sir{\mathtt{SIR}}
\def\pcf{\mathtt{pcf}}

\allowdisplaybreaks 

\usepackage{setspace}	

\newcommand{\subparagraph}{} 
\usepackage{titlesec}        

\begin{document}

\title{Throughput and Age of Information in a Cellular-based IoT Network
}
	\author{Praful D. Mankar, Zheng~Chen,  Mohamed A. Abd-Elmagid,   Nikolaos~Pappas, and Harpreet S. Dhillon
	\thanks{P. D. Mankar is with SPCRC, IIIT Hyderabad, India (Email: praful.mankar@iiit.ac.in). M. A. Abd-Elmagid and H. S. Dhillon are with Wireless@VT, Department of ECE, Virginia Tech, Blacksburg, VA 24061, USA (Email: \{maelaziz,\ hdhillon\}@vt.edu). Z. Chen and N. Pappas are respectively with the Departments of EE, and Science and Technology, Link\"{o}ping University, Link\"{o}ping, Sweden (Email: \{zheng.chen,\ nikolaos.pappas\}@liu.se).
The work of M. A. Abd-Elmagid and H. S. Dhillon has been supported by the U.S. NSF (Grants CNS-1814477 and CPS-1739642). The work of N. Pappas has been supported by the Center for Industrial Information Technology (CENIIT), the Excellence Center at Link\"{o}ping-Lund in Information Technology (ELLIIT), and the Swedish Research Council (VR). This paper is submitted in part to IEEE Globecom 2021 \cite{Praful_ICC}.}
		
	}
	
	\maketitle
	
	\thispagestyle{empty}

\begin{abstract}	
This paper studies the interplay between device-to-device (D2D) communications and real-time monitoring systems in a cellular-based Internet of Things (IoT) network. In particular, besides the possibility that the IoT devices communicate directly with each other in a D2D fashion, we consider that they frequently send time-sensitive information/status updates (about some underlying physical processes observed by them) to their nearest cellular base stations (BSs).
Specifically, we model the locations of the IoT devices as a bipolar Poisson Point Process (PPP) and that of the BSs as another independent PPP.
For this setup, we characterize the performance of D2D communications using the average network throughput metric whereas the performance of the real-time applications is quantified by the Age of Information (AoI) metric. The IoT devices are considered to employ a distance-proportional fractional power control scheme while sending status updates to their serving BSs. {Hence, depending upon the maximum transmission power available, the IoT devices located within a certain distance from the BSs can only send status updates.} This association strategy, in turn, forms the {\em  Johnson-Mehl (JM)} tessellation, such that the IoT devices located in the {\em JM cells} are allowed to send status updates. The average network throughput is obtained by deriving the mean success probability for the D2D links. On the other hand, the temporal mean AoI of a given status update link can be treated as a random variable over space since its success delivery rate is a function of the interference field seen from its receiver. Thus, in order to capture the spatial disparity in the AoI performance, we characterize the spatial moments of the temporal mean AoI. In particular, we obtain these spatial moments by deriving the moments of both the conditional success probability and the conditional scheduling probability for status update links.
Our results provide useful design guidelines on the efficient deployment of future massive IoT networks that will jointly support D2D communications and several cellular network-enabled real-time applications.
\end{abstract}
\IEEEpeerreviewmaketitle
\section{Introduction}
{ With the deployment of a massive number of devices, IoT networks are envisioned to enable a plethora of real-time applications involving features like local decision making and/or remote monitoring and control using some sensory mechanisms. For example, IoT networks can play a vital role in the efficient detection and management of natural disasters  by deploying multiple sensors over a large area (potentially observing multiple physical processes).  In such a scenario, some designated aggregator sensors (or cluster heads) may process the locally collected information from the nearby sensors and forward timely updates to some central unit through cellular BSs for further processing and the subsequent dissemination of the evacuation plans when needed. For such applications, the IoT devices may need to handle different data traffic streams for different destinations, where each stream has different performance objectives, such as rate, latency, or information freshness. }

To account for the heterogeneity of wireless data traffic and multiple functionalities of IoT devices, the interplay between the performance objectives of different data streams becomes an interesting topic. For example, the transmitter of an IoT device can be shared among two different traffic flows, one aiming at maximizing the system throughput by allowing direct D2D communication, and the other one related to monitoring some events in the environment. The D2D communications between nearby IoT devices can be useful for the efficient utilization of their available limited energy sources. For instance, the spatial correlation in the data measurements collected by nearby devices can be exploited for performing their communication tasks in a cooperative manner, thereby reducing the total energy required to execute these tasks. On the other hand, the IoT devices may frequently generate status updates regarding some stochastic processes being observed and send them to the BSs. The objective of updating the information status is to keep the information as fresh as possible, which can be characterized by the Age of information (AoI) \cite{abd2018role,kaul2012real}.
Motivated by the interplay between different IoT applications, we develop a novel analytical framework that allows for a comprehensive analysis of the large-scale IoT networks while integrating both the throughput-oriented D2D traffic and the age-oriented traffic from IoT
devices to BSs into a unified network design.
\subsection{Related Work}
We utilize the concept of AoI to quantify the freshness of information at the BSs regarding random processes monitored by IoT devices \cite{abd2018role}. The authors of \cite{kaul2012real} first introduced AoI for a simple queuing-theoretic model and derived a closed-form expression for the temporal mean (average over infinite period of time) of AoI. Using this result, it was demonstrated in \cite{kaul2012real} that the optimal rate at which the source should generate its update packets in order to minimize the average AoI is different from the optimal rates that either maximize throughput or minimize delay. A series of works then focused on extending the results of \cite{kaul2012real} by characterizing the temporal mean of AoI or other age-related metrics for different variations of queue disciplines \cite{kosta2017age_mono}. These early queuing-theoretic works have inspired the use of AoI or similar age-related metrics to quantify the freshness of information in a variety of communication networks that deal with time sensitive information including, D2D communications \cite{buyukates2019age,D2d_caching,bastopcu2020information} and IoT networks \cite{gu2019timely,abd2018average,AbdElmagid2019Globecom_b,zhou2018joint,Stamatakis_2020,abd2019tcom,AbdElmagid_joint,li2019minimizing,li2020age_a,AbdElmagid2019Globecom_a,Hasan2020,abbas2020joint,ferdowsi2020neural,wang2020minimizing}. The interplay/trade-off between throughput and AoI was also investigated in \cite{ABedewy2016,kadota_throughput,Kosta_Nikos,Gopal_Kaul} for wireless networks with heterogeneous traffic. The prime objective of the works in \cite{buyukates2019age,D2d_caching,bastopcu2020information,gu2019timely,abd2018average,AbdElmagid2019Globecom_b,zhou2018joint,Stamatakis_2020,abd2019tcom,AbdElmagid_joint,li2019minimizing,li2020age_a,AbdElmagid2019Globecom_a,Hasan2020,abbas2020joint,ferdowsi2020neural,wang2020minimizing,ABedewy2016,kadota_throughput,Kosta_Nikos,Gopal_Kaul} was to obtain optimal transmission policies that minimize the temporal mean of AoI or some other age-related metrics for fixed network topologies, referred to as AoI-optimal polices, by
applying different tools from optimization theory.

While the aforementioned works provide a thorough understanding of the temporal statistics of AoI, they are fundamentally limited in their ability to provide insights about the spatial disparity in the AoI performance that is inherently present in wireless networks. This is primarily because each receiver perceives a different signal and interference environments, which cannot be studied using approaches considered in the above works. Once the spatial dimension is explicitly modeled, we can argue that the temporal mean of a performance metric (seen over the complete ensemble of the fading gains), such as transmission rate, delay, or AoI, observed by a receiving device becomes a location-specific quantity. This, in turn, introduces the spatial disparity in the quality of service (QoS) experienced by various wireless links spread across the network. Therefore, it is important to accurately model the spatial distribution of wireless devices to analyze the impact of spatial variations on the achievable QoS.

In recent years, stochastic geometry has emerged as a powerful tool for modeling the spatial distribution of wireless nodes. Most of the early works in this area have focused on characterizing the coverage probability (equivalently, the fraction of devices for which the received signal-to-interference-and-noise ratio (SINR) is above some predefined threshold) in a variety of wireless network settings, such as cellular networks \cite{AndBacJ2011}, heterogeneous networks \cite{DhiGanJ2012} and ad-hoc networks \cite{Baccelli_Aloha2006}.
While this spatio-temporally averaged coverage probability provides useful insight into the network design, it is not sufficient to study the spatial disparity in the link-level performance of the network as discussed above. To overcome this shortcoming, the distribution of location-specific successful transmission probability, termed {\em meta distribution}, was recently introduced in \cite{Haenggi_Meta} to infer useful information like {\em ``the percentage of devices in the network experiencing success probability above some threshold for a given predefined ${\rm SINR}$ value"}. In particular, the moments of the meta distribution were derived for the  bipolar Poisson network in \cite{Haenggi_Meta}, and for the Poisson cellular networks in \cite{Haenggi_Meta} and \cite{Wang_2018}.
However, these stochastic geometry based models are usually agnostic to the traffic variations since they mostly rely on the assumption of saturated queues, i.e., each wireless node always has information to transmit whenever it is scheduled to access the channel. To relax this assumption and allow the traffic aware performance analysis of cellular networks, a semi-analytical framework was developed in
\cite{Blaszczyszyn} and \cite{blaszczyszyn2016spatial} by combining tools from queueing theory (for transmission scheduling) and stochastic geometry (for modeling spatial dimension and hence signal propagation). Further, \cite{Abishek_BirthDeathProcess} studied the spatial birth-death process of randomly arriving wireless links while capturing their stochastic interactions in both space (through interference) and time (through random traffic). A quick glance through the analyses of  \cite{Blaszczyszyn,blaszczyszyn2016spatial,Abishek_BirthDeathProcess} is sufficient to realize that the spatio-temporal performance analysis of of wireless networks is challenging because of: i) the interference-induced correlation between the evolution of queues associated with the transmitting devices, and ii) the temporal variation of the interference field seen by a receiving devices resulting from the stochastic transmission scheduling policy of the transmitting devices.

It is worth noting that the prime focus of the works in  \cite{Blaszczyszyn,blaszczyszyn2016spatial,Abishek_BirthDeathProcess} was on performing the spatio-temporal analysis of conventional performance metrics such as transmission rate and delay. On the other hand, the application of stochastic geometry to perform the spatio-temporal analysis of AoI has been only considered in a handful of recent works \cite{hu2018age,yang2020optimizing,mankar2020stochastic,emara2019spatiotemporal}.
In particular, the authors of \cite{hu2018age,yang2020optimizing,mankar2020stochastic} presented the spatio-temporal analysis of AoI in the context of D2D networks by modeling the D2D links as a bipolar PPP. Specifically, they derived bounds on the spatio-temporal mean AoI \cite{hu2018age}, the spatio-temporal mean peak AoI \cite{yang2020optimizing}, and the spatial distribution of the temporal mean peak AoI \cite{mankar2020stochastic}, by incorporating system modifications to deal with the issue of correlated queues. Besides, the authors of \cite{emara2019spatiotemporal} derived the spatio-temporal mean peak AoI in the context of cellular-based IoT networks while modeling the locations of the BSs and the IoT devices using independent PPPs. Note that since the works in \cite{hu2018age,yang2020optimizing,emara2019spatiotemporal} were focused on characterizing the spatio-temporal mean of AoI or peak AoI, their analyses did not account for the spatial AoI disparity.  In contrast to these works that considered AoI as the only performance quantifying metric, this paper presents a joint spatio-temporal analysis of AoI and throughput for cellular-based IoT networks with heterogeneous traffic as discussed next.

\subsection{Contributions} 
We present a novel stochastic geometry-based analysis of the cellular-based IoT networks which includes: i) the D2D communications between IoT devices, and ii) the transmission of status updates from the IoT devices to the BSs regarding some independent random processes they are sensing. 
{ Each BS is assumed to schedule the transmission of a status update uniformly at random from  one of its associated devices, while the other devices  (i.e., the ones that are not scheduled for status updates) are considered to transmit  regular D2D messages at a fixed rate using  Aloha protocol. }
The locations of the IoT devices are assumed to follow a bipolar PPP whereas the locations of the BSs follow an independent PPP. To improve the delivery rate of the status update transmissions, we assume that each IoT device employs a power control method which is also an important aspect of uplink communications. Further, we consider a generalized system setup wherein the transmission of status updates from the IoT devices within a certain distance from their serving BSs is allowed, leading to the JM tessellation based topology of cellular-based IoT networks as will be formally defined in Section \ref{sec:SysModel} {(please refer to \cite{Priyo_2019_FPR} for more details)}. This is particularly useful to capture the fact that the maximum transmission power of IoT devices is limited in practice. This construction will allow us to account for the correlation between the locations of IoT devices with status updates and the locations of their serving BSs. For this setup, we employ AoI and transmission rate as the key metrics for characterizing the performance of the status update links and D2D links, respectively. The contributions of this paper are briefly summarized below.
\begin{enumerate}
    \item The mean success probability for the D2D links and the moments of the conditional success probability for the status update links are derived.
    \item Moments of the scheduling probability of a status update link are derived while assuming that each BS schedules its associated IoT devices uniformly at random.
    \item Next, we derive the achievable transmission rate for the typical D2D link using its mean success probability.  Further, the spatial moments of the temporal mean AoI of the status update links are derived using the moments of the conditional success probability and scheduling probability.
    \item Our simulation results verify the analytical findings. Next, using numerical results, we highlight the impact of the power control on the achievable D2D network throughput and the spatio-temporal mean AoI for different system design parameters.
\end{enumerate}

To the best of our knowledge, this  paper is the first to develop a joint stochastic geometry-based analysis of AoI and throughput for cellular-based  IoT networks while capturing the spatial disparity in the AoI performance of the status update links.
\section{System Model}
\label{sec:SysModel}
We consider a  cellular-based IoT network wherein the IoT devices can exchange messages in a D2D fashion and also send status updates regarding some random processes to their associated BSs. 
The D2D links of IoT devices are assumed to be randomly distributed according to a homogeneous bipolar PPP wherein the transmitting IoT devices form a PPP $\Phi_{\rm d }$ with intensity $\lambda_{\rm d}$. Their designated receiving IoT devices are independently located at distance $R_{\rm d}$ in uniformly  random directions. The locations of the BSs are also assumed to follow an independent homogeneous PPP $\Phi_{\rm b}$ with intensity $\lambda_{\rm b}$. 

The status updates from the IoT devices contain timestamped measurements of their associated random processes observed in their vicinity. To support  variety of real-time applications, the IoT devices are generally deployed to monitor different types of physical random processes. Therefore,  we assume that the  random processes associated with different IoT devices are independent of each other. The power control is an important aspect of the uplink transmissions in cellular networks for achieving improved transmission rates. 
Therefore, we assume that the IoT devices send status updates to their nearest BSs using a distance-proportional fractional power control scheme.
Specifically, the IoT device at distance $R_{\rm b}$ from its serving BS transmits the status update with power  $P=p_{\rm b}R_{\rm b}^{\alpha\epsilon}$ where $p_{\rm b}$ is the baseline transmit power, $\epsilon\in[0,1]$ is the power control fraction, and $\alpha$ is the path-loss exponent.  Note that $\epsilon=0$ corresponds to the fixed power transmission case (i.e., IoT devices transmit at the fixed power $p_{\rm b}$), and $\epsilon=1$ corresponds to the full power control case (i.e., BSs receive the signals at the fixed power $p_{\rm b}$). The transmission from the devices with high serving link distances naturally require high transmission powers  which may not be possible when the transmission power is limited. For instance, the transmissions of the  devices with serving link distances greater than $\ncalJ=(P_{\rm  \max}/p_{\rm b})^\frac{1}{\alpha\epsilon}$   may fail  when the maximum available transmission power is $P_{\rm \max}$. Therefore,  we consider that the cellular-based status update links can be supported for the IoT devices within distance $\ncalJ$ from their serving BSs using this power control scheme.  As a result, the IoT devices associated with a given BS at ${\bf x}\in\Phi_{\rm b}$ must lie within the intersection $\ncalV_\x=\ncalB_\x(\ncalJ)\cap V_\x$, where  $\ncalB_\x(\ncalJ)$ is the ball of radius $\ncalJ$ centred at $\x$ and $V_\x$ is the Poisson Voronoi (PV) cell which is given by
$$V_{\bf x}=\{{\bf y}\in\mathbb{R}^2:\|{\bf x}-{\bf y}\|\leq \|{\bf z}-{\bf y}\|,  {\bf z}\in \Phi_{\rm b}\}.$$
The set of collection of cells $\{\ncalV_\x\}_{\x\in\Phi_{\rm b}}$ forms a JM tessellation  \cite{moller_1992}. This JM cell based construction provides an attractive way of clustering the mobile users based on their performance in a random geometric setting. For example,  the authors of \cite{Priyo_2019_FPR} applied a similar construction to differentiate between the cell center and cell edge users in the cellular networks.   
\begin{table*}
\centering
{\caption{ Summary of notations }
\label{table:Syatem_Variable}
{\small \begin{tabular}{ |c |c|c|c| }
\hline
 Point processes of BSs and IoT devices   & $\Phi_{\rm b}$ and $\Phi_{\rm d}$  & Transmission rate of D2D links & ${\rm T_d}$ \\ \hline
 BS and IoT device densities & $\lambda_{\rm b}$ and $\lambda_{\rm d}$ &  Cond. mean AoI & $\Delta(\y,\Phi)$ \\ \hline
 Cellular and D2D link distances   & $R_{\rm b}$ and $R_{\rm d}$ & $n$-th moment of cond. mean  AoI & $\Delta_n$ \\ \hline
 Radius of JM cell & $\ncalJ$ & Success prob. of D2D link & ${\rm P_d}$\\  \hline
 PV and JM cells of BS  $\mathbf{x}$ & $V_{\x}$ and $\ncalV_{\x}$& Cond. success prob. of update link & ${\rm P_b}(\y,\Phi)$\\  \hline
 D2D link and uplink baseline tx. powers     & $p_{\rm d}$ and $p_{\rm b}$ &   Moment of Cond. success prob. & $M_b$\\ \hline
 Maximum transmission power of device   & $P_{\rm max}$ & Mean JM cell area & $\bar{\ncalV}_o^1$ \\  \hline
 D2D link medium access probability       & $q_{\rm d}$  & Second moment of JM cell area  & $\bar{\ncalV}_o^2$ \\  \hline
Path-loss exponent & $\alpha$  &  Number of users in JM cell $\ncalV_o$ & $N_{\ncalV_o}$ \\ \hline
 Power control fraction & $\epsilon$   & Cond. update scheduling prob. & $\zeta_{\rm b}(\y,\Phi)$\\ \hline
 $\sir$ thresholds & $\beta_{\rm b}$ and $\beta_{\rm d}$ &  Transmission probability of D2D message & $q_{\rm d}$ \\ \hline
Channel bandwidth & ${\rm B}$  &  & \\ \hline
\end{tabular}}}
\end{table*} 
\subsection{ Transmission Scheduling}
{ The IoT devices are assumed to transmit information packets, containing either regular messages or status updates, in a synchronized time-slotted manner over the same frequency. Thus, the considered system provides co-channel access (or, underlay transmission) for the D2D and cellular-enabled status update links. We will also provide the analysis for the orthogonal channel access (or, overlay transmission) where the types of links (D2D and status updates) are assumed to communicate over orthogonal frequency bands.}
Each BS is assumed to schedule its associated IoT devices for the status update transmission in a uniformly random fashion  to avoid the intra-cell interference. 
{Such a random scheduling policy allows for mathematical tractability and is also meaningful from the perspective of fair resource allocation.}
 { To ensure the timely  delivery of status updates, the devices are assumed to give higher priority to the  status update transmissions over the regular  message transmissions.} Thus, the IoT devices transmit their status updates whenever they are scheduled by their associated BSs. Further, we consider that the IoT devices that are not scheduled for the status updates choose to transmit regular packets on D2D links with probability $q_{\rm d}$ in a given time slot to alleviate the inter-D2D-link interference. Fig. \ref{fig:System_Model} shows a representative realization of the system model discussed above.  
\begin{figure}[h]
\centering
 \includegraphics[clip, trim=0cm 0cm 0cm .4cm, width=.48\textwidth]{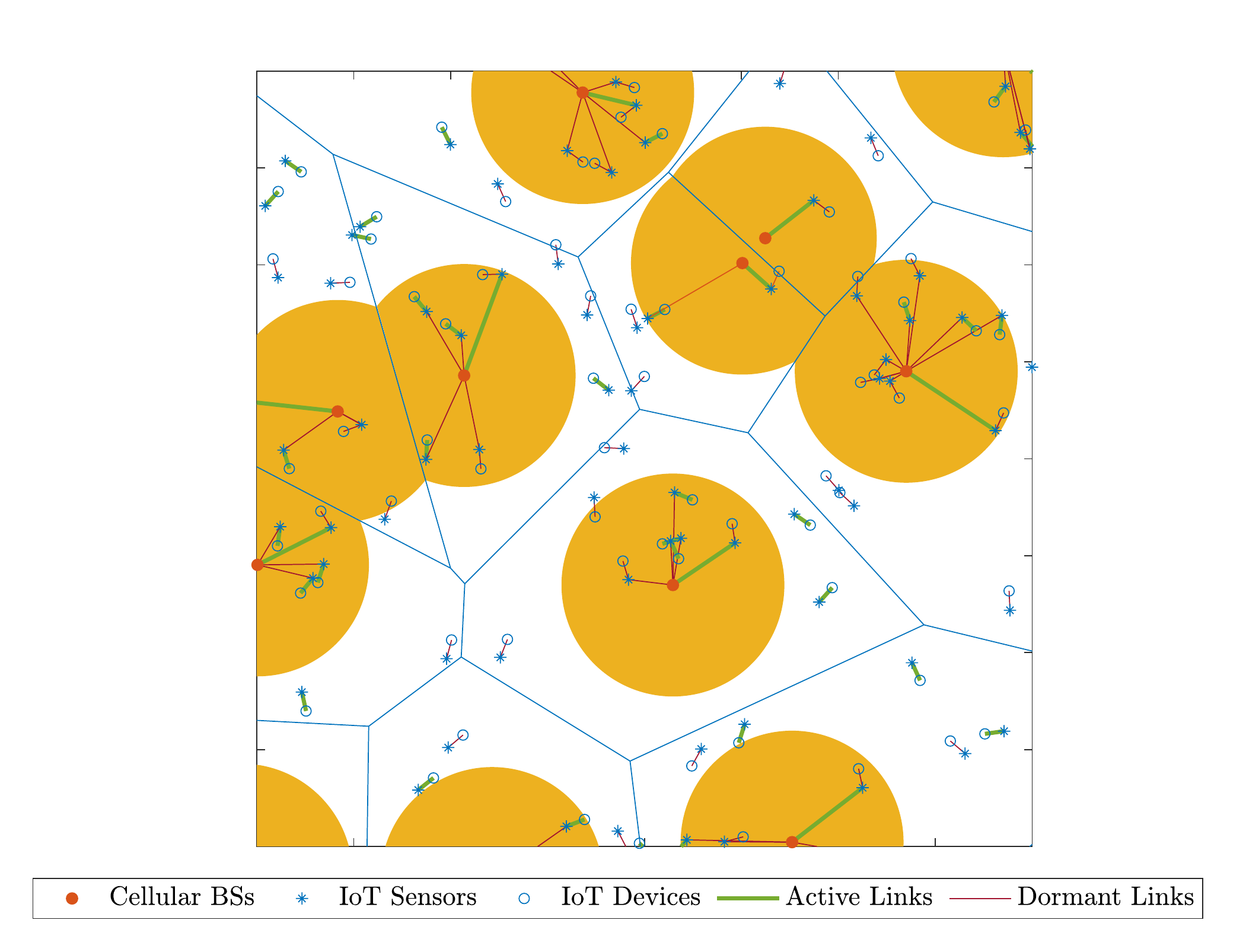}
 \caption{A typical realization of the cellular-based IoT network.}
 \label{fig:System_Model}
\end{figure}
\subsection{Signal-to-Interference Ratio}
Let $\Psi_{\rm b}\subseteq\Phi_{\rm d}$ and $\Psi_{\rm d}\subseteq\Phi_{\rm d}$ denote the sets of the locations of active IoT devices transmitting  status updates and regular D2D messages, respectively. Note that { $\Psi_{\rm b}\cap\Psi_{\rm d}=\emptyset$}. By this construction, we have 
\begin{align*}
\Psi_{\rm b}=\{U(\ncalV_{\bf x}\cap\Phi_{\rm d}):{\bf x}\in\Phi_{\rm b}\},
\end{align*} 
where $U(A)$ represents a point selected uniformly at random from set $A$. We assume $\lambda_{\rm d}\gg\lambda_{\rm b}$ to avoid $\ncalV_{\bf x}\cap\Phi_{\rm d}=\emptyset$  for $\forall {\bf x}\in\Phi_{\rm b}$ with a high probability. This assumption is quite suitable for the IoT network as it requires cellular connectivity to massive number of sensors deployed in the field. 
From Slivnyak's theorem, we know that conditioning on a point of PPP at ${\bf x}$ is equivalent to adding the point ${\bf x}$ to the PPP. Therefore, without loss of generality, we can place the typical BS of the BS PPP $\Phi_{\rm b}\cup\{o\}$  at the origin $o$ and thus the PV cell $V_{o}$ (or JM cell $\ncalV_o$) represents the typical cell in the tessellation. Further, using the stationarity of PPP, we can deduce that the points within the set $\ncalV_{\bf x}\cap\Phi_{\rm d}$ are uniformly distributed in $\ncalV_{\bf x}$. Thus, we will focus our AoI analysis for an updating  device that is distributed uniformly at random in $\ncalV_o$.  This construction facilitates the AoI analysis of the status updates from the perspective of the typical BS which is significantly different than the perspective of the typical IoT device which is expected to reside in the bigger PV (or JM) cells (refer to \cite{mankar2020TypicalCell} for more details). 

Let ${\bf y}\sim U(\ncalV_o)$ denote the location of an IoT device scheduled for the status update transmission, and $R_{\rm b}=\|{\bf y}\|$  denote its distance from  the  typical BS placed at $o$.
We consider the interference-limited scenario. The  signal-to-interference ratio (${\rm SIR}$) received at the typical BS on the status update link from the IoT device at $\y$ is 
\begin{align*}
&{\rm SIR_{\rm b}}=\frac{h_\y R_{\rm b}^{\alpha(\epsilon-1)}p_{\rm b}}{I_{\rm b} },
\end{align*}
where
$$I_{{\rm b}}=\sum_{{\bf x}\in\Phi_{\rm d}} h_{\bf x}\|{\bf x}\|^{-\alpha}\left[p_{\rm d}\mathds{1}({\bf x}\in\Psi_{\rm d}) + p_{\rm b}D_{\x}^{\alpha\epsilon}\mathds{1}({\bf x}\in\tilde{\Psi}_{\rm b})\right],$$ 
$\tilde{\Psi}_{\rm b}=\Psi_{\rm b}\setminus\{\bf y\}$, $p_{\rm d}$ denotes the fixed  power of regular message transmissions on the D2D links, $D_\x$ denotes the distance of the IoT device at $\x$ from its serving BS, and  $h_{\bf x}$ denotes the fading coefficient associated with the link from the IoT device at ${\bf x}$. Assuming independent Rayleigh fading, we model $\{h_{\bf x}\}$ as independent unit mean exponential random variables.

Similar to the typical BS viewpoint discussed above, we perform the D2D network throughput analysis from the perspective of the typical designated  receiving IoT device placed at $o$ by including an additional  transmitting IoT device at $\z\equiv(R_{\rm d},0)$ (paired with the typical designated receiver) to the PPP $\Phi_{\rm d}$. Thus, the ${\rm SIR}$ received at this typical designated IoT receiver becomes
\begin{align*}
&{\rm SIR_{\rm d}}=\frac{h_\z R_{\rm d}^{-\alpha}p_{\rm d}}{I_{\rm d}},
\end{align*}    
where
$$I_{{\rm d}}=\sum_{{\bf x}\in\Phi_{\rm d}} h_{\bf x}\|{\bf x}\|^{-\alpha}[p_{\rm d}\mathds{1}({\bf x}\in\Psi_{\rm d}) + D_\x^{\alpha\epsilon}p_{\rm b}\mathds{1}({\bf x}\in\Psi_{\rm b})].$$
\subsection{Performance Metrics}
\label{subsec:Performance_Metrics}
For the system setting discussed above, our focus is on characterizing  the transmission rate for the typical D2D link  and the spatial disparity in the AoI performance metric measured at the  BSs. 
We assume that the D2D links employ a fixed rate transmission strategy (also termed {\em outage strategy} \cite{Biglieri}) and have saturated queues (i.e., the devices always have a packet to transmit).  The transmission rate of the typical D2D link is
\begin{align}
{\rm T_d}={\rm B}\zeta_{\rm d}\log_2(1+\beta_{\rm d}){\rm P_d},
\label{eq:D2D_TxRate}
\end{align}
where $\zeta_{\rm d}$ and ${\rm P_d}$ are the fraction of transmission time and the successful transmission probability of the typical D2D  link, respectively, and  ${\rm B}$ is the channel bandwidth. 

For the status update transmission, the IoT devices are assumed to generate/sample status updates using {\em generate-at-will} policy \cite{abd2018role}. This policy implies that a  device generates a fresh status update  for the transmission when it is scheduled. Hence, this policy does not require the ACK/NACK protocol or retransmissions since it always transmits a fresh status update   regardless of whether the previous transmission was successful or not.
We employ AoI to characterize the performance of the timely delivery of the status updates from the IoT devices to their  BSs.
The AoI of  status updates received at a BS is defined by the time elapsed from the generation of the latest received status update \cite{kaul2012real}.  Thus,  the AoI measured by the BS related to the status updates from its associated device placed at ${\bf y}\in \ncalV_o$ during time slot $k$ is  
\begin{equation}
A_{\bf y}(k)=k-S_{{\bf y},k},
\end{equation}
where $S_{{\bf y},k}$ is the timestamp of the generation of the  latest received update from the device ${\bf y}$ before time slot $k$.  
 Since the status updates are generated just before their transmissions, the AoI  drops to one whenever a successful transmission occurs. 

The temporal mean AoI of status updates from the device $\y\in\ncalV_o$  that is measured by the typical BS solely depends on its scheduling probability $\zeta_{\rm b}(\y,\Phi)={N_{\ncalV_o}}^{-1}$  and  successful transmission probability ${\rm P_b}(\y,\Phi)$ where $\Phi=\Phi_{\rm d}\cup\Phi_{\rm b}$ and ${N_{\ncalV_o}}$ is the number of devices in $\ncalV_o$. 
Unlike the transmission rate metric given in \eqref{eq:D2D_TxRate}, the AoI is a {{\em nonlinear function}} of these probabilities  (which will be evident in Section \ref{sec:AoI_Througput}).  Therefore, the knowledge of the joint distribution of these conditional probabilities is essential to analysis the spatial distribution of the temporal mean AoI.   
 {For the exact joint analysis of the success probability and scheduling probability for the typical device at $\y\in \ncalV_o$, the key step is to derive the distribution of the area of $\ncalV_o$ given $\y\in \ncalV_o$. However, it is reasonable to  deduce that this exact analysis will be challenging since even the distribution of the area of the typical cell $V_o$ (which is  a much simpler case) is empirically determined \cite{tanemura2003statistical}. In addition,   analyzing scheduling probability jointly with the conditional success probability  will introduce additional complexity.  Therefore, we will derive the scheduling probability of the device at $\y$ by relaxing the  condition $\y\in \ncalV_o$ and perform the AoI analysis under the following widely accepted assumption (e.g., please refer to \cite{Zhong_2017,ElSawy_CelluarIoT_2017,Praful_NOMA,gharbieh2017spatiotemporal}). 
    \begin{assumption}
\label{assumption}
The  cell load $N_{\ncalV_o}$ and the conditional successful transmission probability ${\rm P_d}(\y,\Phi)$ are   independent of each other.
\end{assumption}}
{ In order to verify Assumption 1, we compare simulation results  of the distribution of  conditional (temporal) mean  AoI obtained through the  Monte-Carlo simulations with joint and independent (i.e., Assumption 1) samplings of $N_{\ncalV_o}$ and ${\rm P_b}(\y,\Phi)$.  As will be derived in Section V, the conditional mean AoI of user at $\y$ is given by $\Delta(\y,\Phi)=\frac{N_{\ncalV_o}}{{\rm P_d}(\y,\Phi)}$. Fig. \ref{fig:Assumption_1}  provides a visual verification of the  accuracy of Assumption 1 for the AoI analysis using simulation results. }
\begin{figure}[h]
\centering
\includegraphics[width=.4\textwidth]{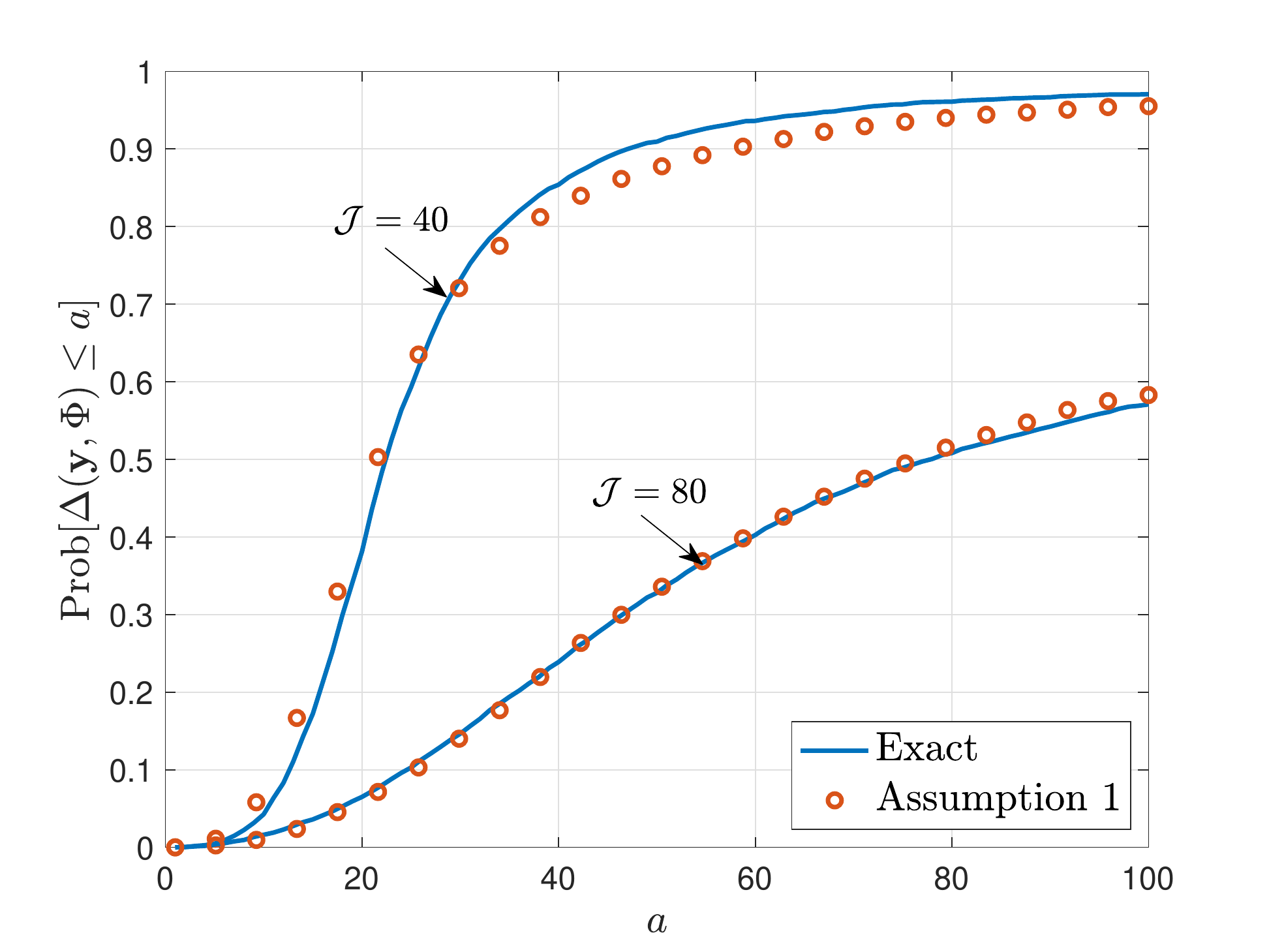}
\caption{Distribution of conditional mean AoI for $\lambda_{\rm b}=10^{-4}$, $\lambda_{\rm d}=40\lambda_{\rm b}$, $\alpha=4$, $\epsilon=0$, and $q_{\rm d}=0$.}
\label{fig:Assumption_1}
\end{figure}
Now, we present the analysis of success probabilities of the  update and  D2D links in the following section which will be used to derive the AoI and D2D network throughput   in Section \ref{sec:AoI_Througput}.
\section{Success Probability Analysis}
\label{sec:SuccessProb}
In this section, we first derive the  success probability for the   regular message transmissions over D2D links. 
Next, we present the analysis of the distribution of conditional success probability for the  status update transmissions over device-BS links. 
\subsection{Success Probability of D2D  Transmission}
\label{subsec:SucceProb_d}
The probability of successful transmission of a regular message for the typical designated D2D receiver can be determined as   
\begin{align*}
{\rm P_d}&=\P[\sir_{\rm d}>\beta_{\rm d}],\\
&=\P\left[h_\z>\beta_{\rm d}R_{\rm d}^\alpha I_{{\rm d}} /p_{\rm d}\right],\\
&=\E_{I_{{\rm d}}}\left[\exp\left(-\beta_{\rm d}R_{\rm d}^\alpha I_{{\rm d}}/p_{\rm d}\right)\right].
\end{align*}

 As $I_{\rm d}$ is the aggregate interference generated from the transmissions of regular messages and status updates, we require the joint distributions of point processes $\Psi_{\rm d}$ and $\Psi_{\rm b}$ to derive the success probability ${\rm P_d}$. However, the exact characterization of this joint distribution is challenging  because $\Psi_{\rm b}$ further depends on the BS PPP $\Phi_{\rm b}$.  
Since $\Psi_{\rm b}$ has exactly one device residing in each cell $\ncalV_\x$, one can interpret $\Psi_{\rm b}$ as the  dependent thinning of the PPP $\Phi_{\rm d}$ for given $\Phi_{\rm b}$. Despite this dependent thinning, the process of the remaining points in  $\Phi_{\rm d}^\prime=\Phi_{\rm d}\setminus\Psi_{\rm b}$ can be closely approximated using a homogeneous PPP with density $\lambda_{\rm d}^\prime=\lambda_{\rm d}-\lambda_{\rm b}$ because of the assumption $\lambda_{\rm d}\gg\lambda_{\rm b}$. Thus,  $\Psi_{\rm d}$ can be directly interpreted as the unconditional thinning of $\Phi_{\rm d}^\prime$ with probability $p_{\rm d}$, hence $\Psi_{\rm d}$ can be modeled  as a PPP with density $q_{\rm d}\lambda_{\rm d}^\prime$.

Besides, the exact characterization of $\Psi_{\rm b}$  is difficult because of the  dependent thinning mentioned above.
On the other hand, one can observe that the density of $\Psi_{\rm b}$ can be approximated with  $\lambda_{\rm b}$  as $\Psi_{\rm b}$ contains exactly  one device in each $\ncalV_\x$. In fact, we have observed that the complementary cumulative distribution function ($\cdf$) of distance from a fixed point, say $o$, to the nearest point in $\Psi_{\rm b }$ closely follows $\exp(-\pi\lambda_{\rm b}r^2)$ which is the {\em void probability} of  BS point process $\Phi_{\rm b}$. Thus,  it is reasonable to approximate  $\Psi_{\rm b}$ with a homogeneous PPP of density $\lambda_{\rm b}$. 
Based on the above observations and to aid the analytical tractability, we consider that the point processes $\Psi_{\rm b }$ and $\Psi_{\rm d}$ are independent of each other. The net interference power received at the typical receiver can be segregated as $I_{{\rm d}}=I_{\Psi_{\rm d}} + I_{{\Psi}_{\rm b}}$
where 
$$I_{\Psi_{\rm d}}=\sum_{{\bf x}\in\Psi_{\rm d}}h_{\bf x}\|{\bf x}\|^{-\alpha}p_{\rm d} \text{~and~}I_{\Psi_{\rm b}}=\sum_{{\bf x}\in\Psi_{\rm b}}h_{\bf x}\|{\bf x}\|^{-\alpha}D_{\x}^{\alpha\epsilon}p_{\rm b}.$$
Since $\Psi_{\rm d}$ and $\Psi_{\rm b}$ are considered to be independent, we can evaluate the success probability as 
\begin{align*}
{\rm P_d}&={\cal L}_{I_{\Psi_{\rm d}}}(\beta_{\rm d}R_{\rm d}^\alpha/p_{\rm d}){\cal L}_{I_{\Psi_{\rm b}}}(\beta_{\rm d}R_{\rm d}^\alpha/p_{\rm d}),\numberthis\label{eq:SuccessProba_LT}
\end{align*}
where ${\cal L}_{X}(\cdot)$ is the Laplace transform (LT) of random variable $X$. The LT of $I_{\Psi_{\rm d}}$ is
\begin{align*}
{\cal L}_{I_{\Psi_{\rm d}}}(s)&=\E_{\Psi_{\rm d}}\prod_{{\bf x}\in\Psi_{\rm d}}\E_{h_{\bf x}}\exp\left(-s p_{\rm d}h_{\bf x}\|{\bf x}\|^{-\alpha}\right),\\
&=\E_{\Psi_{\rm d}}\prod_{{\bf x}\in\Psi_{\rm d}}{\frac{1}{1+s p_{\rm d}\|{\bf x}\|^{-\alpha}}},
\end{align*}
where the first equality follows due to the assumption of independent  fading coefficients.  
Further, using the probability generating functional ($\pgfl$) of the PPP  $\Psi_{\rm d}$, we can obtain
\begin{align*}
{\cal L}_{I_{\Psi_{\rm d}}}(s)&=\exp\bigg(-2\pi q_{\rm d}\lambda_{\rm d}^\prime \int_{0}^\infty \frac{1}{1+(sp_{\rm d})^{-1}r^\alpha}r{\rm d}r \bigg),\\
&=\exp\bigg(-\pi q_{\rm d}\lambda_{\rm d}^\prime \frac{(sp_{\rm d})^\delta}{\sinc(\delta)} \bigg),\numberthis\label{eq:LT_Pd_Id}
\end{align*}     
  where $\delta=\frac{2}{\alpha}$. Now, we obtain the LT of $I_{\Psi_{\rm b}}$ as  
\begin{align*}
{\cal L}_{I_{\Psi_{\rm b}}}(s)&=\E_{\Psi_{\rm b},D_\x}\prod_{{\bf x}\in\Psi_{\rm b}}\E_{h_{\bf x}}\exp\left(-s h_{\bf x}p_{\rm b}D_\x^{\alpha\epsilon}\|{\bf x}\|^{-\alpha}\right),\\
&=\E_{\Psi_{\rm b},D_\x}\prod_{{\bf x}\in\Psi_{\rm b}}\frac{1}{1+s p_{\rm b}D_\x^{\alpha\epsilon}\|{\bf x}\|^{-\alpha}}.
\end{align*}
Recall that $D_\x$ denotes the device-BS link distance, i.e., the distance from the device (with status update) to nearest BS. The link distance  $D_\x$ is naturally smaller than $\ncalJ$ since the devices associated with BS $\x$ are essentially  located within $\ncalV_\x$. 
Therefore, the probability density function ($\pdf$) of the link distance $D_\x$ of a randomly selected device $\x$ becomes  
\begin{align}
f_{D_\x}(u)=\frac{1}{F(\ncalJ)} 2\pi \lambda_{\rm b}u\exp(-\pi \lambda_{\rm b}u^2), \label{eq:Distance_Distribution}
\end{align}
for $0\leq u\leq \ncalJ$
where $F(\ncalJ)=1-\exp(-\pi \lambda_{\rm b}\ncalJ^2)$. Thus
\begin{align*}
{\cal L}_{I_{\Psi_{\rm b}}}(s)&=\E_{\Psi_{\rm b}}\prod_{{\bf x}\in\Psi_{\rm b}}F(\ncalJ)^{-1}\int_0^\ncalJ \frac{2\pi \lambda_{\rm b}u\exp(-\pi \lambda_{\rm b}u^2)}{(1+s p_{\rm b}u^{\alpha\epsilon}\|{\bf x}\|^{-\alpha})}{\rm d}u.
\end{align*}    
Next, using $\pgfl$ of PPP approximation of $\Psi_{\rm b}$, we get ${\cal L}_{I_{\Psi_{\rm b}}}(s)$
\begin{align*}
&=\exp\left(-\lambda_{\rm b}\int_{\R^2} \left[1-\int_0^\ncalJ \frac{2\pi \lambda_{\rm b}u\exp(-\pi \lambda_{\rm b}u^2)}{F(\ncalJ)(1+s p_{\rm b}u^{\alpha\epsilon}\|\x\|^{-\alpha})}{\rm d}u\right]{\rm d}\x\right),\\
&=\exp\left(-\lambda_{\rm b}\int_{\R^2} \int_0^\ncalJ \frac{2\pi \lambda_{\rm b}u\exp(-\pi \lambda_{\rm b}u^2)}{F(\ncalJ)(1+(s p_{\rm b})^{-1}u^{-\alpha\epsilon}\|\x\|^{\alpha})}{\rm d}u{\rm d}\x\right),\\
&=\exp\left(-2\pi\lambda_{\rm b}\int_0^\infty \int_0^\ncalJ \frac{2\pi \lambda_{\rm b}uv\exp(-\pi \lambda_{\rm b}u^2)}{F(\ncalJ)(1+(sp_{\rm b})^{-1}u^{-\alpha\epsilon}v^{\alpha})}{\rm d}u {\rm d}v\right),\\
&=\exp\left(-\frac{\pi\lambda_{\rm b}}{F(\ncalJ)}\frac{(sp_{\rm b})^\delta}{\sinc(\delta)}\int_0^\ncalJ 2\pi \lambda_{\rm b}u^{1+2\epsilon}\exp(-\pi \lambda_{\rm b}u^2){\rm d}u\right),\\
&=\exp\left(-\frac{\pi\lambda_{\rm b}}{F(\ncalJ)}\frac{(sp_{\rm b})^\delta}{\sinc(\delta)}\frac{\gamma(1+\epsilon,\pi \lambda_{\rm b}\ncalJ^2)}{(\pi\lambda_{\rm b})^\epsilon}\right),\numberthis\label{eq:LT_Pd_Ib_1}
\end{align*}
where $\gamma(\cdot,\cdot)$ is a lower incomplete gamma function. 
Finally,  by substituting the LTs of both $I_{\Psi_{\rm d}}$ (given in \eqref{eq:LT_Pd_Id}) and $I_{\Psi_{\rm b}}$ (given in \eqref{eq:LT_Pd_Ib_1}) at $s=\beta_{\rm d}R_{\rm d}^\alpha/p_{\rm d}$ in 
\eqref{eq:SuccessProba_LT}, we obtain the   success probability of regular transmission as presented in the following theorem.
\begin{thm} 
\label{thm:SuccessProba_pd}
For a given $\epsilon$, the success probability of the typical D2D link is ${\rm P_d}=$
  \begin{align*}
\exp\bigg(-\pi q_{\rm d}\lambda_{\rm d}^\prime \frac{\beta_{\rm d}^\delta R_{\rm d}^2}{\sinc(\delta)} -\pi\lambda_{\rm b}\frac{(\beta_{\rm d}p_{\rm b}/p_{\rm d})^\delta R_{\rm d}^2}{\sinc(\delta)}\frac{\gamma(1+\epsilon,\pi \lambda_{\rm b}\ncalJ^2)}{(\pi\lambda_{\rm b})^\epsilon F(\ncalJ)}\bigg).\numberthis\label{eq:SuccessProba_pd}
\end{align*}
\end{thm}
For no power control, i.e., $\epsilon=0$, \eqref{eq:SuccessProba_pd} is simplified  in the following lemma.
\begin{cor}
For $\epsilon=0$, the success probability of the typical D2D link is
  \begin{align*}
{\rm P_d}&=\exp\bigg(-\pi q_{\rm d}\lambda_{\rm d}^\prime \frac{\beta_{\rm d}^\delta R_{\rm d}^2}{\sinc(\delta)} -\pi \lambda_{\rm b} \frac{(\beta_{\rm d}p_{\rm b}/p_{\rm d})^\delta R_{\rm d}^2}{\sinc(\delta)} \bigg).\numberthis\label{eq:SuccessProba_pd_epi_0}
\end{align*}
\end{cor}
\begin{proof}
For $\epsilon=0$, ${\rm P_d}$ given in \eqref{eq:SuccessProba_pd_epi_0}  follows by substituting    $\gamma(1,x)=1-\exp(-x)$ in \eqref{eq:SuccessProba_pd}.
\end{proof}
{\begin{cor}
Under orthogonal access, the success probability of the typical D2D link  is
  \begin{align*}
{\rm \tilde{P}_d}&=\exp\bigg(-\pi q_{\rm d}\lambda_{\rm d}^\prime \frac{\beta_{\rm d}^\delta R_{\rm d}^2}{\sinc(\delta)} \bigg).\numberthis\label{eq:SuccessProba_pd_ortho}
\end{align*}
\end{cor}
\begin{proof}
The proof follows by setting the density $\lambda_{\rm b}$ of interfering update links to zero in \eqref{eq:SuccessProba_pd}.
\end{proof}}
\subsection{Success Probability of the Status Update Transmission}
\label{subsec:SucceProb_a}
The success probability of the status update transmission is defined as the probability that ${\rm SIR_b}$ is above a threshold $\beta_{\rm b}$. Similar to the analysis presented in Section \ref{subsec:SucceProb_d}, this  success probability can be derived by averaging over the space. However, this spatially averaged success probability is not very useful to characterize the performance of non-linear metrics, such as AoI, as will be evident in Section \ref{sec:AoI_Througput}. For this reason, the distribution of the conditional success probability, termed {\em meta distribution} \cite{Haenggi_Meta}, is required.  Since the meta distribution is difficult to determine directly \cite{Haenggi_Meta},  our first goal is to   derive its moments. 
Given $\Phi=\Phi_{\rm d}\cup\Phi_{\rm b}$, the conditional success probability of status update from the IoT device at ${\y}\in \ncalV_o$  is 
\begin{align*}
{\rm P_b}(\y,\Phi)&=\P[{\rm SIR_b}>\beta_{\rm b}|\Phi]=\exp\left(-\beta_{\rm b}R_{\rm b}^{\alpha(1-\epsilon)} I_{\rm b}/p_{\rm b} \right).
\end{align*} 

While $\y$ is already included in $\Phi$, we explicitly condition ${\rm P_b}$ on $\y$ to indicate that the IoT device at $\y$ is scheduled for the status update transmission. 
Given $\Phi$, the conditional success probability depends on the evolution of the point process $\Phi_{\rm d}$ whose  devices are randomly scheduled for the status update and D2D message transmissions. However, given the complexity of characterizing  point  process of interfering devices (transmitting status updates) even for a fixed time instance, as presented in \cite{Priyo_2019_FPR,user-point}, it is reasonable to presume that the exact characterization of evolution of $\Phi_{\rm d}$ is even more challenging. Therefore, we perform the conditional success probability analysis while considering the interference powers received from the IoT devices  transmitting regular messages and status updates are independent across the transmission slots. 
Therefore, it is safe to assume that devices scheduled for status update transmissions and regular message transmissions are drawn from independent point processes. 

Since each BS is assumed to schedule its associated users uniformly at random,  the probability that an IoT device at $\z\in \ncalV_{\x}$ transmits the status update  in a given slot is 
\begin{align} 
\label{eq:StatusUpdate_Link_scheduling_Probabilityf}
 \zeta_{\rm b}(\z|\Phi)=N_{\ncalV_\x}^{-1},\text{~for~}\x\in\Phi_{\rm b},
 \end{align}
 where $N_{\ncalV_\x}$ is the number of IoT devices in set $\Phi_{\rm d}\cap\ncalV_\x$.   
The IoT devices that are not scheduled for status update transmission are assumed to transmit regular messages with probability $q_{\rm d}$. 
 Hence, we consider that the IoT device $\z\in\Phi_{\rm d}$ transmits regular messages with probability 
\begin{align*}
\zeta_{\rm d}&=q_{\rm d} \P\left[\z\notin \bigcup\nolimits_{\x\in\Phi_{\rm b}} \ncalV_\x\right]\\
&+ q_{\rm d}\E\left[(1-\zeta_{\rm b}(\z|\Phi))|\z\in \bigcup\nolimits_{\x\in\Phi_{\rm b}} \ncalV_\x\right]\P\left[\z\in \bigcup\nolimits_{\x\in\Phi_{\rm b}} \ncalV_\x\right],\\
&=q_{\rm d} F(\ncalJ)+ q_{\rm d}(1-\zeta_{\rm b})(1-F(\ncalJ)),\numberthis\label{eq:D2D_Link_scheduling_Probability}
\end{align*}  
where $\zeta_{\rm b}=\E[\zeta_{\rm b}(\z|\Phi))|\z\in \bigcup\nolimits_{\x\in\Phi_{\rm b}}]$. { The  scheduling  probability $\zeta_{\rm b}$ can be obtained using the probability mass function (${\rm pmf}$) of $N_{\ncalV_o}$ which will be derived in Lemma \ref{lemma:NoDevice_pmf}.}

{ As discussed above, we approximate the locations of devices transmitting regular messages and status updates using independent point processes and denote them by $\Omega_{\rm d}$ and $\Omega_{\rm b}$, respectively. }
Thus,  the conditional success probability can be written as
\begin{align*}
{\rm P_b}(\y,\Phi)&=\prod_{{\bf x}\in\tilde{\Omega}_{\rm b}}\bigg(\frac{\zeta_{\rm b}}{1+\beta_{\rm b} R_{\rm b}^{\alpha(1-\epsilon)} D_\x^{\alpha\epsilon} \|{\bf x}\|^{-\alpha}} +1-\zeta_{\rm b}\bigg)\\
&~~~\times\prod_{{\bf x}\in\Omega_{\rm d}}\bigg(\frac{\zeta_{\rm d}}{1+\beta_{\rm b} R_{\rm b}^{\alpha(1-\epsilon)}\|{\bf x}\|^{-\alpha}\frac{p_{\rm d}}{p_{\rm b}}} + 1-\zeta_{\rm d}\bigg)
\end{align*}   
where $\tilde{\Omega}_{\rm b}=\Omega_{\rm b}\setminus\{\Omega_{\rm b}\cap\ncalV_o\}$. 
 The $b$-th moment of conditional success probability is given by
{\begin{align*}
M_b&=\E_{\y,\Phi}[{\rm P_b}(\y,\Phi)^b]\\
&=\E_{R_{\rm b}}\bigg[\underbrace{\E\prod_{{\bf x}\in\tilde{\Omega}_{\rm b}}\bigg(1-\frac{\zeta_{\rm b}}{1+\beta_{\rm b}^{-1} R_{\rm b}^{\alpha(\epsilon-1)} D_\x^{-\alpha\epsilon}\|{\bf x}\|^{\alpha}}\bigg)^b}_{{\cal A}}\\
&\underbrace{\E\prod_{{\bf x}\in\Omega_{\rm d}}\left(1-\frac{\zeta_{\rm d}}{1+\beta_{\rm b}^{-1} R_{\rm b}^{\alpha(\epsilon-1)} \|{\bf x}\|^\alpha\frac{p_{\rm b}}{p_{\rm d}}}\right)^b}_{{\cal B}}\bigg].\numberthis\label{eq:Mb_1}
\end{align*} }  
Based on the arguments presented in Section \ref{subsec:SucceProb_d}, it is reasonable to assume that  the devices with regular messages follow a homogeneous PPP with density $\lambda_{\rm d}$ and model their medium access probability using $\zeta_{\rm d}$ given in \eqref{eq:D2D_Link_scheduling_Probability}. Therefore, using \cite[Theorem 1]{Haenggi_Meta}, we obtain
\begin{align}
{\cal B}=\exp\bigg(-\pi\lambda_{\rm d}R_{\rm b}^{2(1-\epsilon)}C(b)\bigg),
\end{align} 
where 
\begin{align}
    C(b)=\frac{(\beta_{\rm b}p_{\rm d}/p_{\rm b})^{\delta}}{\sinc(\delta)}\sum_{k=1}^\infty {b\choose k}{\delta-1\choose k-1} \zeta_{\rm d}^k.
    \label{eq:Cb}
\end{align}
On the other hand, to determine the expectation involved in the term ${\cal A}$ of \eqref{eq:Mb_1}, we require the distribution of $\tilde{\Omega}_{\rm b}$ as seen from the typical BS at $o$.  For this, we first charaterize the point process $\tilde{\Psi}_{\rm b}$ which contains the  devices from $\tilde{\Omega}_{\rm b}$ transmitting status updates in a given time slot.  The pair correlation function ($\pcf$) of  this point process of interferers  $\tilde{\Psi}_{\rm b}$  with respect to the BS at $o$ for given $\ncalJ$ is derived in \cite{Priyo_2019_FPR}  as
\begin{align}
g(r;\ncalJ)=1-\exp\left(-2\pi \bar{\ncalV}_{o}^{-1}r^2\right), \text{~for~}  r\geq 0,
\end{align}
where $\bar{\ncalV}_{o}^{-1}=\E[|\ncalV_o|^{-1}]$ and $|A|$ represents the area of set $A$.  The $\pdf$ of  $|\ncalV_o|$ will be derived in  Section \ref{sec:schedulin_probability_AoI} which can be used here to determine $\bar{\ncalV}_{o}^{-1}$.
Further, the authors of  \cite{Priyo_2019_FPR} used this $\pcf$ to approximate $\tilde{\Psi}_{\rm d}$  using a non-homogeneous PPP with density $\lambda_{\rm b}g(r)$.  
However, in our case, the active set of interferers are actually scheduled from $\tilde{\Omega}_{\rm b}$ by their associated BSs such that there is exactly one interfering device in each cell $\ncalV_\x$ at a given time slot. 
Therefore, we can approximate  $\tilde{\Omega}_{\rm b}$ using a non-homogeneous PPP with density $\lambda_{\rm d}D(r;\ncalJ)$ where
\begin{align}
 D(r;\ncalJ)=F(\ncalJ)g(r;\ncalJ),\label{eq:Dg}   
\end{align}
such that the term $F(\ncalJ)$ represents the probability that a device is located in one of the  cells $\ncalV_\x$ for $\x\in\Phi_{\rm b}$. 
Thus, we can interpret  that $\tilde{\Psi}_{\rm b}$ is a result of thinning  $\tilde{\Omega}_{\rm b}$ with scheduling probability $\zeta_{\rm b}$. 
Assuming $D_\x$s to be independent of each other, we can write $ \ncalA=$
  \begin{align*}
 \E_{\tilde{\Phi}_{\rm d}}\prod_{\x\in \tilde{\Omega}_{\rm b}} \int \bigg[1-\frac{\zeta_{\rm b}}{(1+\beta_{\rm b}^{-1} R_{\rm b}^{\alpha(\epsilon-1)} u^{-\alpha\epsilon}\|\x\|^\alpha)^k}\bigg]^bf_{D_\x}(u){\rm d}u.
  \end{align*}
  
 The distribution of distance from the nucleus to a uniformly random point in the typical PV cell follows $1-\exp(-\pi\lambda_{\rm b}{\rm c_1}r^2)$, where ${\rm c_1}=\frac{9}{7}$ \cite[Theorem 3]{mankar2019distance}. Thus, the $\pdf$ of link distance $D_\x$ of device associated with a randomly selected BS can be approximated using \eqref{eq:Distance_Distribution} with corrected density ${\rm c_1}\lambda_{\rm b}$. However, it may be noted that the  link distance $D_\x$ of interfering user $\x$ must be smaller than $\|\x\|$ as it is closer to its serving BS than the typical BS at $o$.
Thus,  using the $\pdf$ of $D_\x$ and the $\pgfl$ of the non-homogeneous PPP approximation of $\tilde{\Omega}_{\rm b}$, we obtain $ \ncalA =$
  \begin{align}
 \exp\left(-4\pi^2{\rm c_1}\lambda_{\rm d}\lambda_{\rm b}\int\nolimits_0^\infty \hspace{-4mm}D(v;\ncalJ) \int\nolimits_0^{\min(v,\ncalJ)}\hspace{-6mm} f(u,v;R_b,b) {\rm d}uv{\rm d}v \right)
  \end{align}
  where $f(u,v;R_b,b)=$
  \begin{align} 
  \left(1-\left[1-\frac{\zeta_{\rm b}\beta_{\rm b} R_{\rm b}^{\alpha(1-\epsilon)} u^{\alpha\epsilon}}{\beta_{\rm b} R_{\rm b}^{\alpha(1-\epsilon)} u^{\alpha\epsilon}+v^\alpha}\right]^b\right)\frac{u\exp(-\pi{\rm c_1}\lambda_{\rm b}u^2)}{F(\sqrt{{\rm c_1}}\min(v,\ncalJ))}.\label{eq:fd}
  \end{align}
Finally, by substituting ${\cal A}$ and ${\cal B}$ in \eqref{eq:Mb_1} and then averaging using the $\pdf$ of serving link distance $R_{\rm b}$ given in \eqref{eq:Distance_Distribution}, we obtain the $b$-th moment of ${\rm P_b}(\y,\Phi)$  in the following theorem.
\begin{thm}
\label{thm:Moment_Cond_SuccessProb}
For given $\epsilon$, the $b$-th moment of the conditional success probability of status update at the typical BS is
\begin{align*}
M_b=&\frac{2\pi {\rm c_1}\lambda_{\rm b}}{F(\sqrt{{\rm c_1}}\ncalJ)}\int_0^\ncalJ r\exp\bigg(-\pi {\rm c_1}\lambda_{\rm b}r^2 -\pi\lambda_{\rm d}\\
&~~~~\left({\cal G}(r,b) + r^{2(1-\epsilon)}C(b)\right)\bigg){\rm d}r,\numberthis\label{eq:SuccessProba_pa}
\end{align*} 
\text{where}~
\begin{align*}
    {\cal G}(r,b)&=4\pi\lambda_{\rm b}{\rm c_1}\int_0^\infty D(v;\ncalJ)\int_0^{\min(v,\ncalJ)}f(u,v;r,b)v{\rm d}u{\rm d}v,\\
\end{align*}
and $C(b)$, $D(v;\ncalJ)$ and $f(u,v;r,b)$ are given by \eqref{eq:Cb}, \eqref{eq:Dg}, and \eqref{eq:fd}, respectively. 
\end{thm}
The following lemma presents simplified expressions for $M_b$ given in Theorem \ref{thm:Moment_Cond_SuccessProb}   for the special cases of no power control and full power control. 
\begin{cor}
\label{lemma:SuccessProba_pa}
The $b$-th moment of the conditional success probability of status update at the typical BS under full power control (i.e., $\epsilon=1$) is
\begin{align*}
M_b=&\exp\left(-\pi\lambda_{\rm d}\left(\tilde{\cal G}(b)+C(b)\right) \right), \numberthis\label{eq:SuccessProba_pa_epi1}
\end{align*} 
where
\begin{align*}
    \tilde{\cal G}(b)=&4\pi\lambda_{\rm b}{\rm c_1}\int_0^\infty D(v;\ncalJ)\int_0^{\min(v,\ncalJ)}f(u,v;1,b){\rm d}uv{\rm d}v,
\end{align*}
and under no power control (i.e., $\epsilon=0$) is
\begin{align*}
M_b=&\frac{2\pi {\rm c_1}\lambda_{\rm b}}{F(\ncalJ)}\int_0^\ncalJ \exp\bigg(-\pi {\rm c_1}\lambda_{\rm b}r^2 \\
&~~~~~~~~-\pi\lambda_{\rm d}\left(\hat{\cal G}(r,b)+r^2C(b)\right) \bigg)r{\rm d}r,\numberthis\label{eq:SuccessProba_pa_epi0}
\end{align*} 
where
\begin{align*}
\hat{\cal G}(r,b)=2\int_0^\infty D(v;\ncalJ)\bigg(1-\bigg[1-\frac{\zeta_{\rm b}\beta_{\rm b}r^{\alpha}}{\beta_{\rm b}r^{\alpha}+v^\alpha}\bigg]^{b}\bigg)v{\rm d}v.
\end{align*}
\end{cor}
\begin{cor}
\label{lemma:SuccessProba_pa_ortho}
Under orthogonal access, the $b$-th moment of the conditional success probability of status update at the typical BS is 
\begin{align}
\tilde{M}_b=&\frac{2\pi {\rm c_1}\lambda_{\rm b}}{F(\ncalJ)}\int_0^\ncalJ r\exp\left(-\pi {\rm c_1}\lambda_{\rm b}r^2 -\pi\lambda_{\rm d}{\cal G}(r,b) \right){\rm d}r,\numberthis\label{eq:SuccessProba_pa_ortho}
\end{align} 
which under full power control and no power control becomes
\begin{align}
\tilde{M}_b=&\exp\left(-\pi\lambda_{\rm d}\tilde{\cal G}(b) \right)
\end{align}
and $\tilde{M}_b=$
\begin{align}
&\frac{2\pi {\rm c_1}\lambda_{\rm b}}{F(\ncalJ)}\int_0^\ncalJ \exp\bigg(-\pi {\rm c_1}\lambda_{\rm b}r^2 -\pi\lambda_{\rm d}\hat{\cal G}(r,b) \bigg)r{\rm d}r,\numberthis\label{eq:SuccessProba_pa_epi10_ortho}
\end{align} 
respectively, where  ${\cal G}(r,b)$ is given in Theorem \ref{thm:Moment_Cond_SuccessProb}, and $\tilde{\cal G}(b)$ and $\hat{\cal G}(r,b)$ are given in Corollary \ref{lemma:SuccessProba_pa}.
\end{cor}
\begin{proof}
The proof follows by setting  $\zeta_{\rm d}=0$ in \eqref{eq:SuccessProba_pa}-\eqref{eq:SuccessProba_pa_epi0}.
\end{proof}

\section{Analysis of Cell Load}
\label{sec:schedulin_probability_AoI}
As discussed in Section \ref{subsec:Performance_Metrics}, the temporal mean AoI seen by a status update link depends  jointly on its ability of successful transmission  and probability of getting scheduled. 
Therefore, in this section, we derive the scheduling probability of the typical IoT device and then use it along with Assumption \ref{assumption} to derive the moments of the conditional mean AoI in Section \ref{sec:AoI_Througput}.  

Recall that each BS is assumed to schedule the status update transmission uniformly at random from one of its associated devices in a given time slot. Thus, the scheduling probability of a device associated with the typical BS placed at $o$ depends on the load of cell $\ncalV_o$ (i.e., number of devices $N_{\ncalV_o}$ located in $\ncalV_o$).  As a result, the scheduling probability of  a device at $\y\in\Phi_{\rm d}\cap\ncalV_o$ for given $\Phi$ is 
$\zeta_{\rm b}(\y,\Phi)=N_{\ncalV_o}^{-1}$.
By the PPP definition, the distribution of number devices located in a region is parameterized by its area. Thus, the knowledge for the area distribution of  $\ncalV_o$ is essential to determine the scheduling probability of a device associated with the typical BS placed at $o$. However, it is difficult to directly derive the area distribution of a random set. Thus, we first determine the moments of area of  $\ncalV_o$ which will then be used to accurately characterize its distribution. While these moments are derived in \cite{Priyo_2019_FPR}, we derive a simplified expression for  the second moment of area of $\ncalV_o$ in Lemma \ref{lemma:Area_Moments} using the approach presented in \cite{Foss1993OnAC}.
\begin{lemma}
\label{lemma:Area_Moments}
For a given $\ncalJ$, the mean of area of the typical cell $\ncalV_o$ is 
\begin{align}
\bar{\ncalV}_o^1=\frac{1}{\lambda_{\rm b}}\left(1-\exp\left(-\pi\lambda_{\rm b}\ncalJ^2\right)\right),
\label{eq:mean_Vo}
\end{align}  
and the second moment of area of the typical cell  $\ncalV_o$ is
\begin{align}
\bar{\ncalV}_o^2&=2\pi\lambda_{\rm b}^{-2}\int_0^\pi\int_{0}^{\pi-u} \frac{G(u,v)}{S(u,v)^2}\bigg[1-\left(1+\lambda_{\rm b}\ncalJ^2S^\prime(u,v)\right)\nonumber\\
&~~~~\exp\left(-\lambda_{\rm b}\ncalJ^2{S}^\prime(u,v)\right)\bigg]{\rm d}v {\rm d}u,
\label{eq:2nd_moment_Vo}
\end{align}
\text{where}~
\begin{align*}
    G(u,v)&=\sin(u)\sin(v)\sin(u+v),\\ ~S^\prime(u,v)&=S(u,v)\max(\sin(u),\sin(v))^{-2},\\
\text{and}~
S(u,v)&=G(u,v) +\left(\pi-v\right) \sin(u)+\left(\pi-u\right)\sin(v).
\end{align*}
\end{lemma}
\begin{proof}
Please refer to Appendix \ref{appendix:Area_Moments} for the proof.
\end{proof}
Let $R_m$ be the half of the distance from the typical BS to its nearest BS. We have    $\ncalV_o=\ncalB_o(\ncalJ)$ whenever the event $\ncalE=\{R_m>\ncalJ\}$ occurs. Thus, the $\pdf$ of the area of $\ncalV_o$ becomes
\begin{align*}
f_{\ncalV_o}(v)=\delta(\pi\ncalJ^2)\P[\ncalE] + f_{\ncalV_o}(v|\ncalE^C)\P[\ncalE^C],
\end{align*}   
where $\delta(\cdot)$ is the  Dirac-delta function and $\P[\ncalE^C]=1-\P[\ncalE]$. From the {\em void probability} of PPP, we get $\P[\ncalE]=\exp(-4\pi\lambda_{\rm b}\ncalJ^2)$. Similar to \cite{Priyo_2019_FPR}, we approximate the distribution $f_{\ncalV_o}(v|\ncalE^C)$ using the truncated beta distribution as 
\begin{align}
f_{\ncalV_o}(v|\ncalE^C)=\frac{v^{\kappa_1-1}(2 \pi\ncalJ^2-v)^{\kappa_2-1}}{(2\pi\ncalJ^2)^{\kappa_1+\kappa_2-1}{\rm B}(\kappa_1,\kappa_2)}, 
\end{align}
 ~\text{for}~$0\leq v\leq \pi\ncalJ^2$, \text{where} ${\rm B}(\kappa_1,\kappa_2)=\int_0^{1/2} v^{\kappa_1-1}(1-v)^{\kappa_2-1}{\rm d}v$.
 Note that the support of the truncated distribution is $[0,\pi\ncalJ^2]$ whereas the support of untruncated distribution is considered to be $[0,2\pi\ncalJ^2]$. 
We determine the parameters $\kappa_1$ and $\kappa_2$ through moment matching method. For this, we obtained the first and second moments of the area of $\ncalV_o$ conditioned on $\ncalE^C$ using Lemma \ref{lemma:Area_Moments} as
\begin{align}
\tilde{\ncalV}_o^1&=\E[|\ncalV_o||\ncalE^C],\nonumber\\
&=\left(\E[\ncalV_o]-\E[\ncalV_o|\ncalE]\P[\ncalE]\right)\P[\ncalE^C]^{-1},\nonumber\\
&=\frac{\bar{\ncalV}_o^1-\pi\ncalJ^2\exp(-4\pi\lambda_{\rm b}\ncalJ^2)}{1-\exp(-4\pi\lambda_{\rm b}\ncalJ^2)},
\end{align}
and
\begin{align}
\tilde{\ncalV}_o^2&=\E[|\ncalV_o|^2|\ncalE^C],\nonumber\\
&=\left(\E[|\ncalV_o|^2]-\E[|\ncalV_o|^2|\ncalE]\P[\ncalE]\right)\P[\ncalE^C]^{-1},\nonumber\\
&=\frac{\bar{\ncalV}_o^2-\pi^2\ncalJ^4\exp(-4\pi\lambda_{\rm b}\ncalJ^2)}{1-\exp(-4\pi\lambda_{\rm b}\ncalJ^2)}.
\end{align}

Therefore, the parameters of approximate truncated beta distribution can be determined by  solving the following simultaneous equations
\begin{align}
\tilde{\ncalV}_o^1&=2 \pi\ncalJ^2\frac{{\rm B}(\kappa_1+1,\kappa_2)}{{\rm B}(\kappa_1,\kappa_2)},\\
\text{~and~}\tilde{\ncalV}_o^2&=(2 \pi\ncalJ^2)^2\frac{{\rm B}(\kappa_1+2,\kappa_2)}{{\rm B}(\kappa_1,\kappa_2)}.
\end{align}

Finally, by substituting the truncated beta approximation of $f_{\ncalV_o}(v|\ncalE)$ in $f_{\ncalV_o}(v)$, we obtain
\begin{align}
&f_{\ncalV_o}(v)=\delta(\pi\ncalJ^2)\exp(-4\pi\lambda_{\rm b}\ncalJ^2)\nonumber\\
&+\frac{1-\exp(-4\pi\lambda_{\rm b}\ncalJ^2)}{(2 \pi\ncalJ^2)^{\kappa_1+\kappa_2-1}{\rm B}(\kappa_1,\kappa_2)}v^{\kappa_1-1}(2\pi\ncalJ^2-v)^{\kappa_2-1},
\label{eq:Area_pdf}
\end{align}
for $0\leq v\leq \pi\ncalJ^2$. The accuracy of the above approximation of  area distribution of  $\ncalV_o$ has been discussed extensively in \cite{Priyo_2019_FPR}. 
Using \eqref{eq:Area_pdf}, we now present the ${\rm pmf}$ of number of IoT devices located in  $\ncalV_o$ in the following lemma, which will be used to analyze AoI in Section \ref{subsec:AoI}.
\begin{lemma}
\label{lemma:NoDevice_pmf}
The ${\rm pmf}$ of the number of devices residing in $\Phi_{\rm d}\cap\ncalV_o$ is 
\begin{align}
\P[N_{\ncalV_o}=n] = \frac{1}{n!}\int_0^{\pi\ncalJ^2}(\lambda_{\rm d}v)^n\exp(-\lambda_{\rm d}v)f_{\ncalV_o}(v){\rm d}v, 
\end{align}
\text{~for~} $n\geq 0$, where $f_{\ncalV_o}(v)$ is given by \eqref{eq:Area_pdf}.
\end{lemma}  

\section{ D2D Throughput and Average AoI }
\label{sec:AoI_Througput}
In this section, we first determine the D2D network throughput using the success probability of  regular message transmissions derived in Theorem \ref{thm:SuccessProba_pd}. Next, we will characterize the spatial distribution of the temporal mean AoI  using the moments of conditional success probability of update transmissions derived in Theorem \ref{thm:Moment_Cond_SuccessProb} and the cell load distribution derived in Lemma \ref{lemma:NoDevice_pmf}. 
\subsection{Throughput of D2D Network}
\label{subsec:Throughput}
The network throughput is measured by the average number of successfully delivered information bits per unit area per second per Hertz (bit/s/Hz/m$^2$). 
Note that the effective probability of an IoT device transmitting the regular messages is $\zeta_{\rm d}$ (refer to \eqref{eq:D2D_Link_scheduling_Probability}). Therefore,  for a given density  $\lambda_{\rm d}$ of the IoT devices, the throughputs of the typical D2D link and the D2D network can be determined as
\begin{align}
{\rm T_d}=\zeta_{\rm d}{\rm B}\log_2(1+\beta_{\rm d}){\rm P_d} \text{~and~} {\rm T_{N}}=\lambda_{\rm d}{\rm T_d},
\end{align}
respectively, where $\zeta_{\rm d}$ is given in \eqref{eq:D2D_Link_scheduling_Probability} and ${\rm P_d}$ is  given in Theorem \ref{thm:SuccessProba_pd}.
\subsection{Spatial Distribution of Temporal Mean AoI } 
\label{subsec:AoI}
In this section, our goal is to derive the spatial distribution of the temporal mean AoI observed by the IoT device-BS links.  
\begin{figure}[h]
\centering
 \includegraphics[clip, trim=1.3cm 14cm 0.5cm 5cm, width=0.5\textwidth]{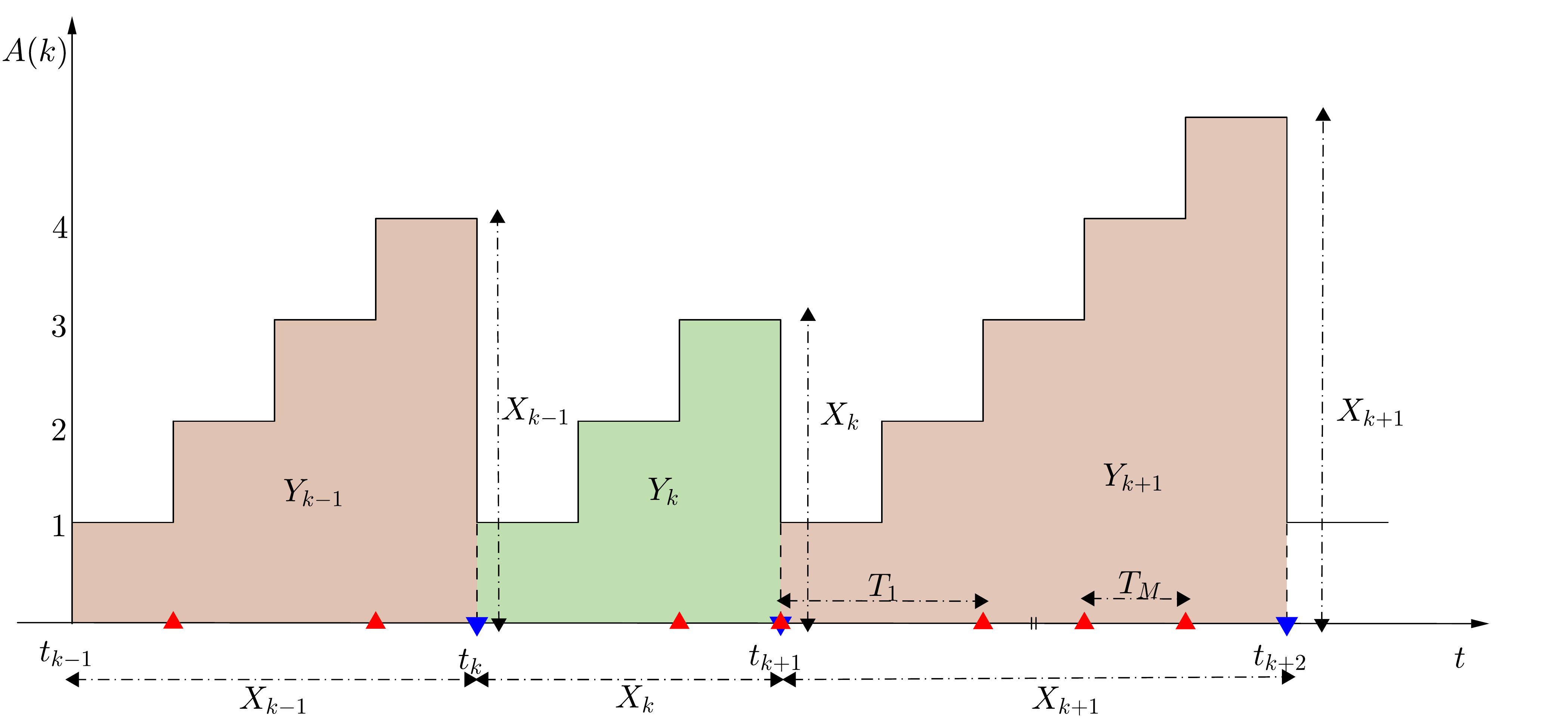}
 \caption{{ Sample path of AoI $A_{\bf y}(k)$ for the IoT device at ${\bf y}$.} The red upward and blue downward arrows show the transmission attempts and successful transmissions, respectively.}
 \label{fig:aoi}
\end{figure}
Fig. \ref{fig:aoi} depicts a representative sample path of the AoI for the system model discussed in Section \ref{subsec:Performance_Metrics}.   { Let $Y_{{\bf y},k}$ and $X_{{\bf y},k}$ denote the sum of AoI $A_{\bf y}(k)$ (i.e., area of shaded region)  and the time difference between the successful reception of the $k$-th and the ($k+1$)-th status updates from device at ${\bf y}$, respectively. Thus, we can write
\begin{equation}
Y_{{\bf y},k}=\sum_{k=t_k}^{t_{k+1}} A_{\bf y}(k) \text{~and~} X_{{\bf y},k}=\sum\limits_{i=1}^{L_{\bf y}}T_{{\bf y},i}, 
\label{eq:X_k}
\end{equation} 
where $T_{{\bf y},i}$ denotes the time elapsed between two consecutive scheduling instances of device ${\bf y}$ and $L_{\bf y}$ denotes the number of attempted transmissions between two successfully received status updates from device ${\bf y}$.} The temporal  mean AoI (for a device-BS link conditioned on $\Phi$) is charaterized  here similarly to \cite{infocom19} wherein the authors determine temporal mean AoI for the case of a single point-to-point  link. { For a period of $N$ time slots, where $K_{\bf y}$ successful updates occur, the temporal mean AoI for device at ${\bf y}$ conditioned on $\Phi$ is 
\begin{align*}
\Delta({\bf y},\Phi;N)&=\frac{1}{N}\sum\limits_{k=1}^{N}A_{\bf y}(k),\\
&=\frac{1}{N}\sum\limits_{k=1}^{K_{\bf y}}Y_{{\bf y},k},\\
&=\frac{K_{\bf y}}{N}\frac{1}{K_{\bf y}}\sum\limits_{k=1}^{K}Y_{{\bf y},k}.\numberthis
\end{align*}
Using  $\lim\limits_{N\rightarrow\infty}\frac{K_{\bf y}}{N}=\frac{1}{\mathbb{E}[X_{\bf y}]} \text{~and~} \lim\limits_{K_{\bf y}\to\infty}\frac{1}{K_{\bf y}}\sum\limits_{k=1}^{K_{\bf y}}Y_{\bf y}=\mathbb{E}[Y_{\bf y}],$
we can obtain the mean AoI for the device ${\bf y}$ for given $\Phi$ as
\begin{equation}
\Delta({\bf y},\Phi)=\lim\limits_{N\rightarrow\infty}\Delta({\bf y},\Phi;N)=\frac{\mathbb{E}[Y_{\bf y}]}{\mathbb{E}[X_{\bf y}]}.
\end{equation}
Further, we can establish the relation between $Y_{{\bf y},k}$ and $X_{{\bf y},k}$ as 
\begin{equation}
Y_{{\bf y},k}=\sum\limits_{m=1}^{X_{{\bf y},k}}m=\frac{1}{2}X_{{\bf y},k}(X_{{\bf y},k}+1).
\end{equation}
Thus, we can obtain 
\begin{equation}
\Delta({\bf y},\Phi)=\frac{1}{2}\frac{\mathbb{E}\left[X_{{\bf y},k}(X_{{\bf y},k}+1)\right]}{\mathbb{E}[X_{\bf y}]}=\frac{\mathbb{E}[X_{{\bf y}}^2]}{2\mathbb{E}[X_{{\bf y}}]}+\frac{1}{2}.
\label{eq:aoi}
\end{equation}
From \eqref{eq:aoi}, it is evident that the knowledge of the first two moments of $X_{{\bf y},k}$ is sufficient  to evaluate the temporal mean of AoI. However, the distribution of $X_{{\bf y},k}$ is not identical for the IoT devices spread across the network for the following reasons. The distribution of $X_{{\bf y},k}$ of an IoT device-BS link jointly depends on its scheduling and successful transmission probabilities. 
In particular, for  a given $\Phi$ and the IoT device at $\y\in V_o$,  the scheduling probability $\zeta_{\rm b}(\y,\Phi)$ and conditional success probability ${\rm P_b}(\y,\Phi)$  charaterize the distributions of  $T_{{\bf y},i}$ and $L_{\bf y}$, respectively, which essentially determine the temporal mean AoI through $X_{{\bf y},k}$.  This implies that the temporal mean AoI observed at an IoT device-BS link is conditioned on the locations of the IoT devices and the BSs. Hence, we refer to this mean AoI as the {\em conditional temporal mean AoI}. Our goal is to derive the spatial distribution of the temporal mean AoI. 
 \subsubsection{Conditional temporal mean AoI}  
For  the IoT device at $\y\in\ncalV_o$ given $\Phi$, the probability of successful transmission of status update  is  ${\rm P_b}(\y,\Phi)$ and the probability that it is scheduled for the status update is $\zeta_{\rm b}(\y,\Phi)$. Therefore, the $\pmf$s of $T_{{\bf y},i}$ and $L_{\bf y}$  become
\begin{align}
\P[T_{{\bf y},i}=t|\Phi]&=\zeta_{\rm b}(\y,\Phi)[1-\zeta_{\rm b}(\y,\Phi)]^{t-1},\label{eq:Tyi}\\
\text{~and~}\P[L_{\bf y}=m|\y,\Phi]&={\rm P_b}(\y,\Phi)[1-{\rm P_b}(\y,\Phi)]^{m-1},\label{eq:Ly}
\end{align}
 for $1\leq m,t$, respectively. Since $T_{{\bf y},i}$s are independent and identically distributed (because of the random scheduling), we can apply the Wald’s identity and obtain the mean of $X_{{\bf y},k}$ as
\begin{equation} \label{EX-gen}
\mathbb{E}[X_{{\bf y}}]=\mathbb{E}[T_{{\bf y}}]\mathbb{E}[L_{\bf y}]= \frac{1}{\zeta_{\rm b}(\y,\Phi){\rm P_b}(\y,\Phi)}.
\end{equation}
Now, we determine the second moment of $X_{{\bf y},k}$. From its definition, we can write
\begin{equation*}
X_{{\bf y},k}^2=\left(\sum\limits_{i=1}^{L_{\bf y}}T_{{\bf y},i}\right)^2=\sum\limits_{i=1}^{L_{\bf y}}T_{{\bf y},i}^2+\sum\limits_{i=1}^{L_{\bf y}}\sum\limits_{j=1,j\neq i}^{L_{\bf y}}T_{{\bf y},i} T_{{\bf y}j}.
\end{equation*}
Note that $T_{{\bf y},i}$ and $T_{{\bf y}j}$, for $i\neq j$, are independent because each BS schedules its associated IoT devices uniformly at random in a given slot. Thus,  for $L_{\bf y}=m$, we get 
\begin{align*}
\mathbb{E}[X_{{\bf y}}^2 \vert L_{\bf y}=m]&=m\mathbb{E}[T_{{\bf y}}^2]+m(m-1)\mathbb{E}[T_{{\bf y}}]^2,\nonumber\\
&=m {\rm Var}[T] + m^2\E[T]^2,\nonumber\\
&=m\frac{1-\zeta_{\rm b}(\y,\Phi)}{\zeta_{\rm b}(\y,\Phi)^2}+m^2\frac{1}{\zeta_{\rm b}(\y,\Phi)^2}
\end{align*}
Now, by averaging over the $\pmf$ of $L_{\bf y}$ given in \eqref{eq:Ly}, we obtain
\begin{align*} 
\mathbb{E}[X_{{\bf y}}^2 ]&=\frac{1-\zeta_{\rm b}(\y,\Phi)}{\zeta_{\rm b}(\y,\Phi)^2}\E[L_{\bf y}]+\frac{1}{\zeta_{\rm b}(\y,\Phi)^2}\E[L_{\bf y}^2]\\
&=\frac{1-\zeta_{\rm b}(\y,\Phi)}{\zeta_{\rm b}(\y,\Phi)^2}\frac{1}{{\rm P_b}(\y,\Phi)}+\frac{1}{\zeta_{\rm b}(\y,\Phi)^2}\frac{2-{\rm P_b}(\y,\Phi)}{{\rm P_b}(\y,\Phi)^2}.\numberthis\label{EX2-gen}
\end{align*}}
Finally, by substituting \eqref{EX-gen} and \eqref{EX2-gen} into \eqref{eq:aoi}, we obtain the conditional temporal mean AoI as given in the following lemma.
\begin{lemma}
\label{lemma:Cond_Mean_AoI}
For a given $\Phi$, the conditional temporal mean AoI measured by the typical BS of the status updates from the IoT device located at $\y\in \ncalV_o$  is
\begin{equation} 
\Delta(\y,\Phi) =  \frac{1}{\zeta_{\rm b}(\y,\Phi){\rm P_b}(\y,\Phi)}.
\label{eq:Cond_Mean_AoI}
\end{equation} 
\end{lemma}

\subsubsection{Spatial Moments of $\Delta(\y,\Phi)$}  
In this subsection, we analyze the spatial distribution of temporal mean AoI under Assumption \ref{assumption}. 
Thus, it is apparent from Lemma \ref{lemma:Cond_Mean_AoI} that the $n$-th moment of $\Delta(\y,\Phi)$ is equal to the product of $n$-th moments of ${\rm P_b}(\y,\Phi)^{-1}$ and $\zeta_{\rm b}(\y,\Phi)^{-1}$  which can be directly obtained from Theorem \ref{thm:Moment_Cond_SuccessProb} and Lemma \ref{lemma:NoDevice_pmf}, respectively. 
\setcounter{theorem}{2}
\begin{theorem}
\label{theorem:Moments_conditional_AoI}
The $n$-th moment of the temporal mean AoI of the status updates generated from the IoT devices is 
\begin{align}
\Delta_{n}=\E[N_{\ncalV_o}^n|N_{\ncalV_o}\geq 1]M_{-n},\label{eq:Moment_AoI_n}
\end{align}
where $M_{-n}$ is given in Theorem \ref{thm:Moment_Cond_SuccessProb} and ${\rm pmf}$ of $N_{\ncalV_o}$ is given in Lemma \ref{lemma:NoDevice_pmf}.
\end{theorem} 
\begin{proof}
Using the assumption of independence of $\zeta_{\rm b}(\y,\Phi)$ and ${\rm P_b}(\y,\Phi)$ and Lemma \ref{lemma:Cond_Mean_AoI}, the $n$-th moment of the conditional temporal mean AoI can be obtained as
\begin{align*}
\Delta_n&=\E_{\y,\Phi}[\Delta(\y,\Phi)^n]=\E_{\y,\Phi}[{\zeta_{\rm b}(\y,\Phi)}^{-n}]\E_{\y,\Phi}[{\rm P_b}(\y,\Phi)^{-n}].
\end{align*}
Thus, we arrive at \eqref{eq:Moment_AoI_n} by plugging the $(-n)$-th moment of ${\rm P_b}(\y,\Phi)$ from Theorem  \ref{thm:Moment_Cond_SuccessProb} and using the $\pmf$ of $N_{\ncalV_o}$ given in Lemma \ref{lemma:NoDevice_pmf}. 
\end{proof}
\begin{figure*}[h]
	\centering
	\hspace{-5mm}\includegraphics[ width=.33\textwidth]{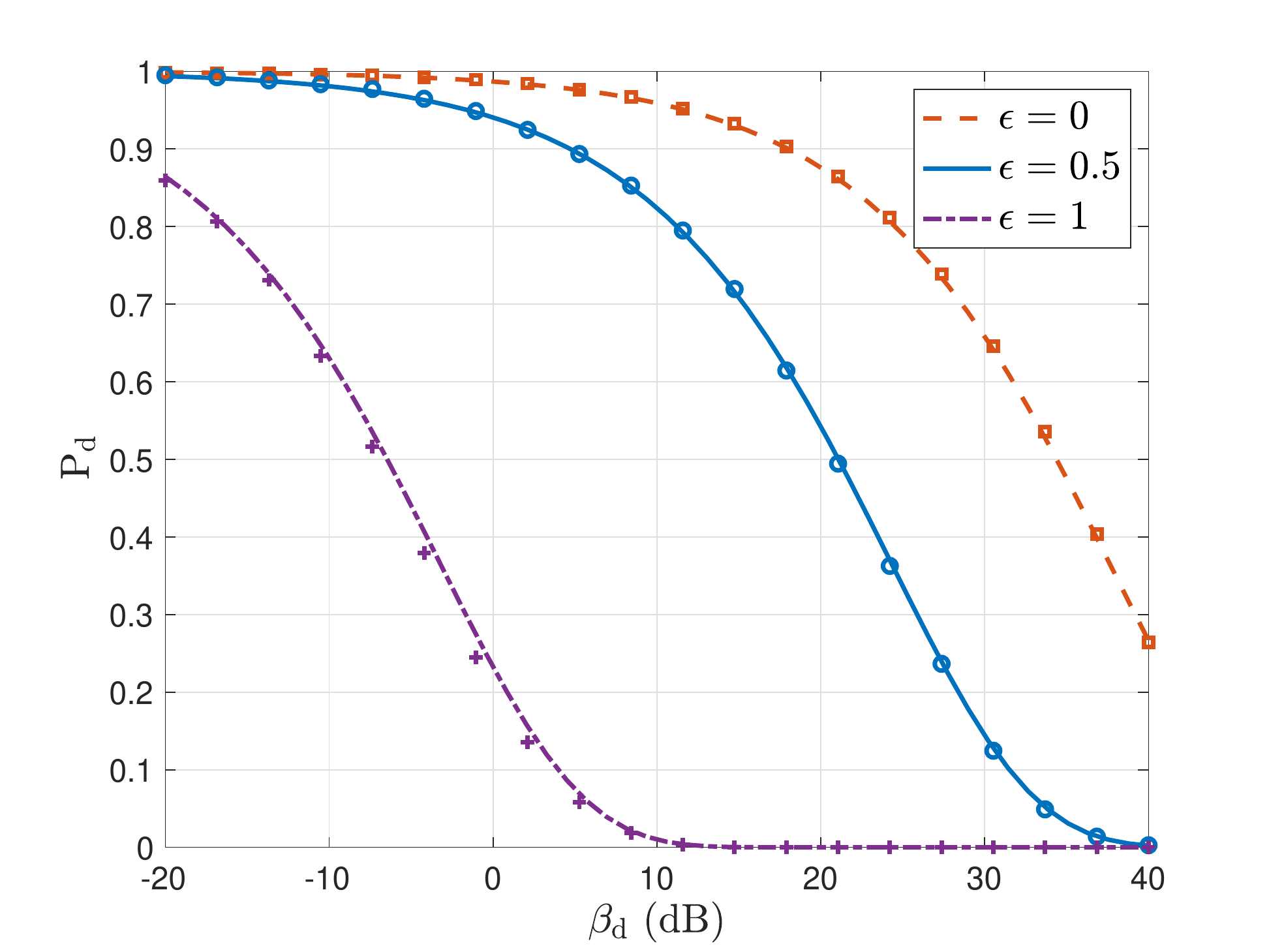}
		\includegraphics[ width=.33\textwidth]{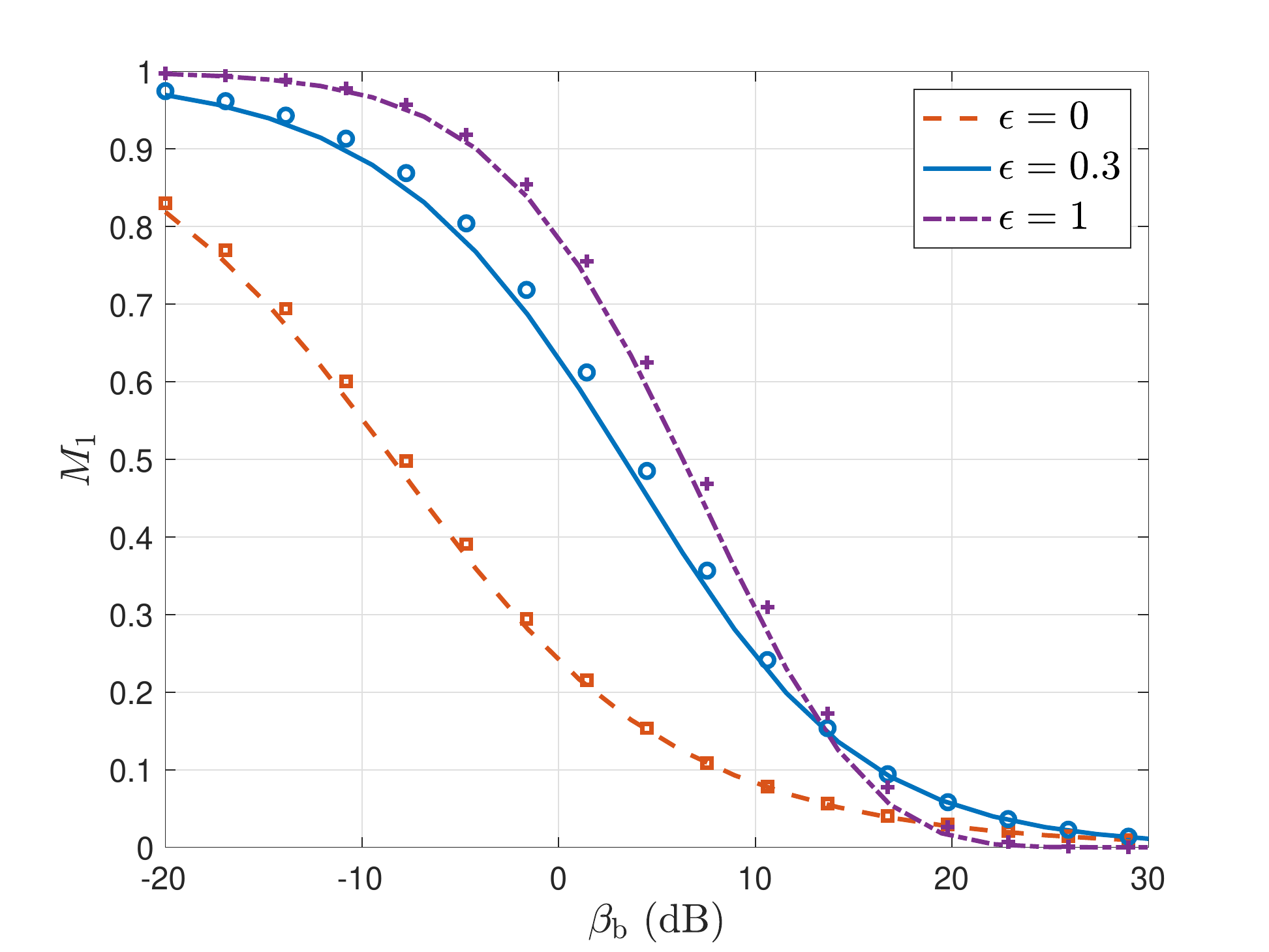}
		\includegraphics[ width=.33\textwidth]{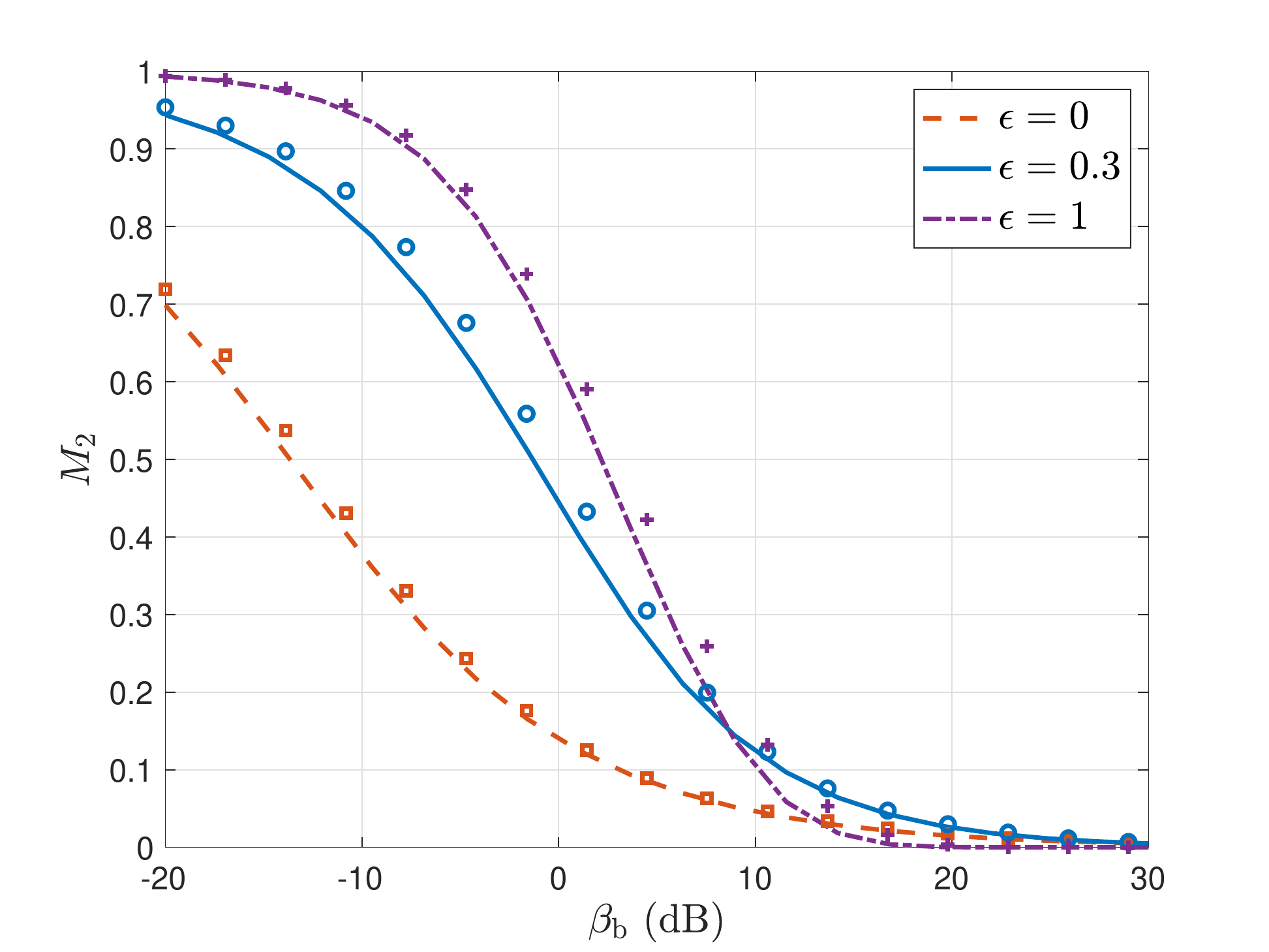}
	\caption{Left: success probability of  regular message transmissions on D2D links. Middle and Right: first and second moments of conditional success probability of status update transmissions.} 
	\label{fig:SuccessProb_Pd}
\end{figure*}
{\begin{cor}
 The spatiotemporal mean of the AoI is 
 \begin{equation}
 \Delta_{1}=\frac{\lambda_{\rm d}}{\lambda_{\rm b}}\left(1-\exp(-\pi{\rm c_1}\lambda_{\rm b}\ncalJ^2)\right)M_{-1}
 \end{equation}
 where $M_{-1}$ is given in Theorem \ref{thm:Moment_Cond_SuccessProb}.
\end{cor}}
Simplified expressions for the moments of the temporal mean AoI can be obtained for the special cases of no power control and full power control using the moments of conditional success probability $M_b$ presented in Corollary \ref{lemma:SuccessProba_pa}. 
In addition, the moments of the temporal mean AoI for the orthogonal access can also be obtained using the  moments of conditional success probability $\tilde{M}_{b}$ presented in Corollary \ref{lemma:SuccessProba_pa_ortho}. They not repeated here due to lack of space.

 {\begin{remark}
   Note that Theorem \ref{theorem:Moments_conditional_AoI} presents the spatial moments of the mean AoI for a general case as it allows to control the status update support for the devices experiencing link quality (which is expected to decrease with the increase of serving link distance) above a certain percentile by appropriately setting $\ncalJ$  (or, $P_{\rm max}$ and $\epsilon$).  The status update support for all devices is a special case to which our analysis can be easily extended by simply setting $\ncalJ=\infty$  (for which, we need $P_{\rm max}=\infty$ or $\epsilon=0$). However,  it may be noted that $M_{-n}$ (thus the spatial moments $\Delta_n$) becomes unbounded as $\ncalJ\to\infty$ which can be verified using  \eqref{eq:SuccessProba_pa}. Therefore, it is important to appropriately select $\ncalJ$ such that it covers the devices of interest. From this perspective, the JM cell based analysis of AoI is meaningful.       
    \end{remark}}

\section{Numerical Analysis and Discussion}
\label{sec:Results}
In this section, we first verify the success probabilities of transmissions of regular messages and status updates derived in Section \ref{sec:SuccessProb}  using simulation results. Next, we will discuss the impact of various system design parameters on our key performance metrics (i.e., D2D network  throughput and AoI associated with status updates) presented in Section \ref{sec:AoI_Througput} using numerical results.  For the numerical analysis,  the system parameters are considered as  $\lambda_{\rm b}=10^{-4}$ BSs/m$^2$, $\lambda_{\rm d}=20\lambda_{\rm b}$ devices/m$^2$, ${\rm B}=200$ KHz, $\ncalJ=40$ m,  $p_{\rm b}=p_{\rm d}=100$ dBm, $\alpha=4$,  $q_{\rm d}=0.3$,  $R_{\rm d}=2$ m, and $\beta_{\rm b}=3$ dB, unless mentioned otherwise. Note that  the JM cell radius ${\cal J}=40$ m  provides coverage to around $40$\% of the IoT devices for the status update transmissions.  {In our simulations, we perform the spatial averaging of temporal mean AoI and conditional success probability over 10000 network realizations and for each realization the temporal averaging (on small scale fading)  is performed over 1000 transmission slots. }
 
 Fig. \ref{fig:SuccessProb_Pd} (left)   verifies the accuracy of the success probability of the  regular message transmissions, and Fig. \ref{fig:SuccessProb_Pd} (middle and right) verifies the accuracy of the first two moments of the conditional success probability of the status update transmissions.  The curves correspond to the analytical results whereas the markers correspond to the simulation results.  
Fig. \ref{fig:SuccessProb_Pd} (middle and right) shows that the power control provides improvement in the success probability of  the status update transmissions. However, it can be observed from the figure that increasing power control fraction $\epsilon$ beyond $0.3$ will not contribute much in the improvement of success probability of status update because it becomes limited by the interference from the regular message transmissions over D2D links. In addition, it is also necessary to  select  a small value of  $\epsilon$ to ensure better success probability of D2D links.   { A smaller $\epsilon$ provides better success probability in the high $\sir$ regime. This is because the devices with higher $\sir$ lie closer to their serving BSs, thus for these devices, the desired signal power received at their BSs does not improve faster with increasing $\epsilon$ compared to the increase of the inter-cell interference.}
 

 Let  $\mathcal{D}_{\rm \lambda}=\frac{\lambda_{\rm d}}{\lambda_{\rm b}}$ represents the ratio of densities of devices and BSs. Fig. \ref{fig:Througput_AoI} (left) shows the impact of $\mathcal{D}_{\rm \lambda}$ and $\epsilon$ on the achievable throughput of D2D network for $R_{\rm d}=$ 2 m and 5 m. The achievable throughputs of D2D link and D2D network are determined as  ${\rm T}^*_{\rm d}=\max_{\beta_{\rm d}} {\rm T}_{\rm d}$ and ${\rm T}_{\rm N}^*=\lambda_{\rm d}{\rm T}_{\rm d}^*$, respectively.
 The initial rise in the achievable D2D link throughput is because of better chances of medium access for regular message transmission (since the update scheduling probability drops with increasing $\mathcal{D}_{\rm \lambda}$). However, the D2D link throughput drops eventually with increasing $\mathcal{D}_{\rm \lambda}$ because of the increased interference. Nevertheless, the achievable D2D network throughput  monotonically increases with $\mathcal{D}_{\rm \lambda}$.   The figure also shows that the achievable  throughput is higher when   D2D communication range is shorter.     { The BS density $\lambda_{\rm b}$ has two interrelated impacts on the D2D network throughput performance: 1) increasing  $\lambda_{\rm b}$  reduces the transmission powers  of the status updating devices (because of the smaller  serving link distances) which positively affects the D2D throughput, and 2) increasing $\lambda_{\rm b}$ leads to higher  density of status updating devices which negatively affects the D2D throughput.  Fig.  5 (right) shows the D2D throughput as a function of $\lambda_{\rm b}$ for a fixed $\lambda_{\rm d}=10^{-2}$ and sufficiently large $\ncalJ$ (such that $\ncalV_o\approx V_o$). The larger value of $\ncalJ$ is selected to see the maximum benefit of increasing $\lambda_{\rm b}$ through the reduced transmission power as stated above.  However, the  figure  reveals  that  the  D2D  throughput  degrades  as  $\lambda_{\rm b}$ increases which in turn implies that the  negative impact is dominant.  }
\begin{figure*}[h]
	\centering
		\includegraphics[ width=.45\textwidth]{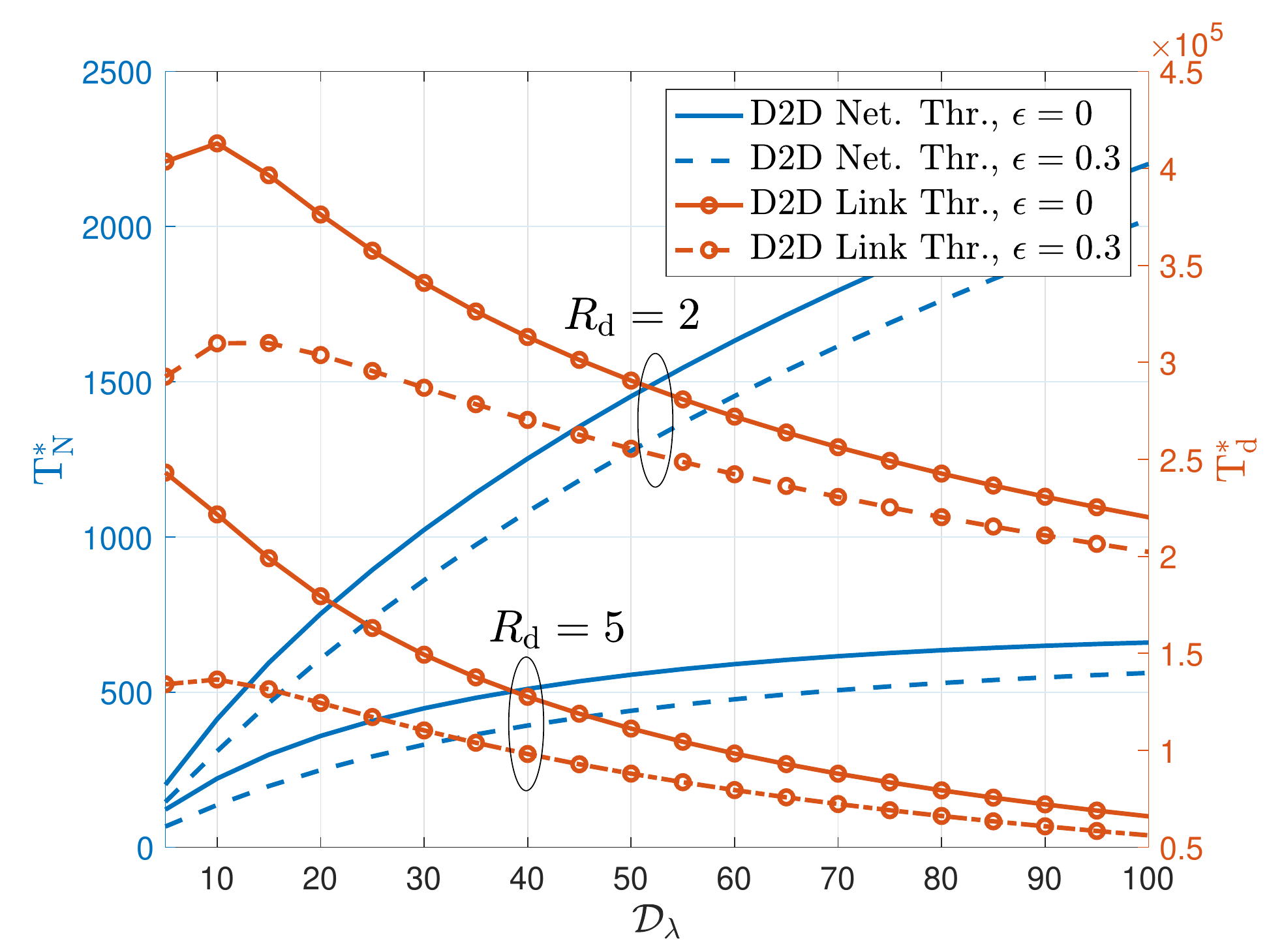}
		\includegraphics[width=.45\textwidth]{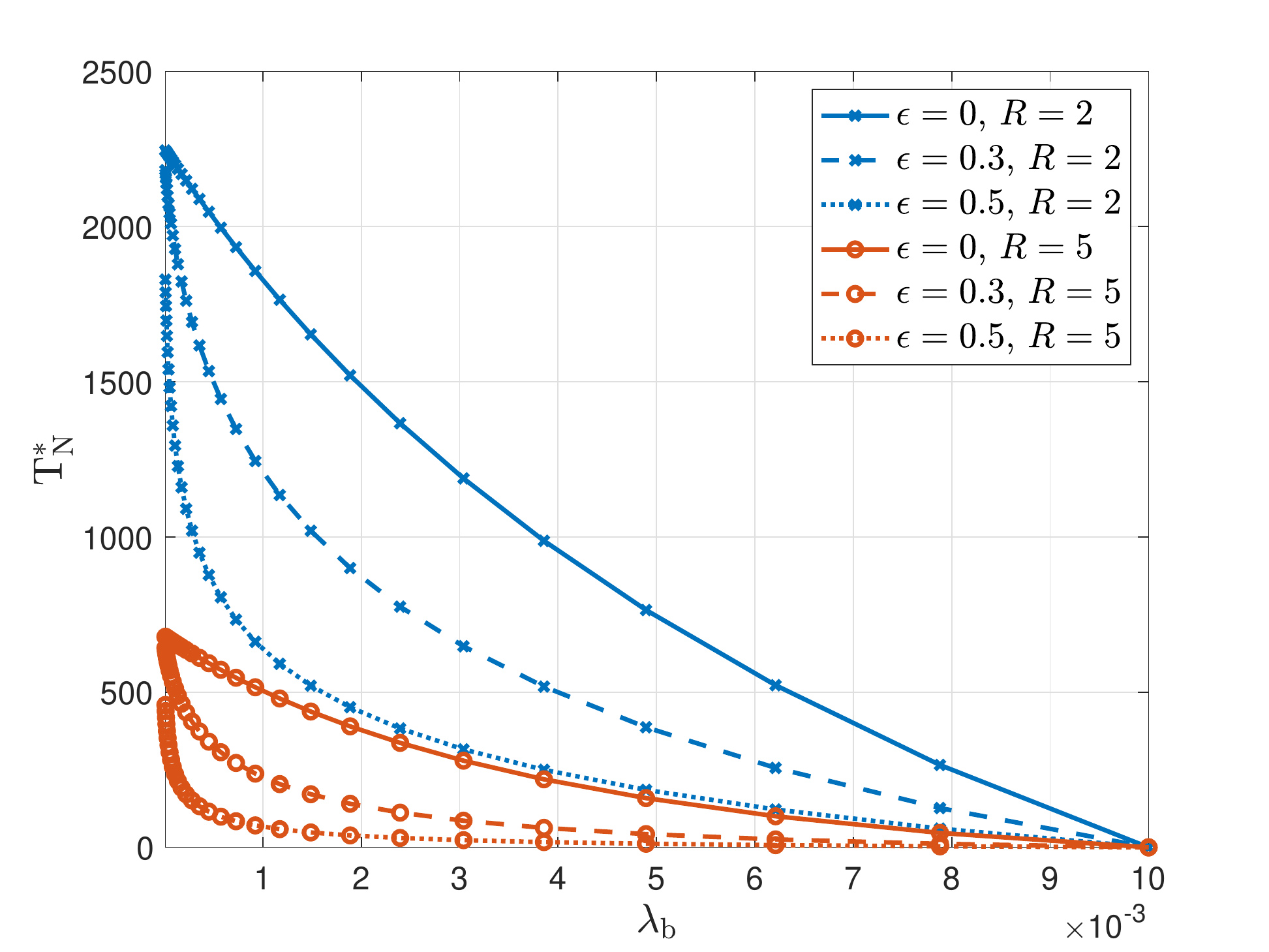}
	\caption{ Left: achievable  D2D network and D2D link throughputs. { Right: achievable D2D network throughput vs  $\lambda_{\rm b}$ for $\lambda_{\rm d}=10^{-2}$.}}
	\label{fig:Througput_AoI}
\end{figure*} 
\begin{figure*}[h]
	\centering
		\hspace{-4mm}	\includegraphics[ width=.33\textwidth]{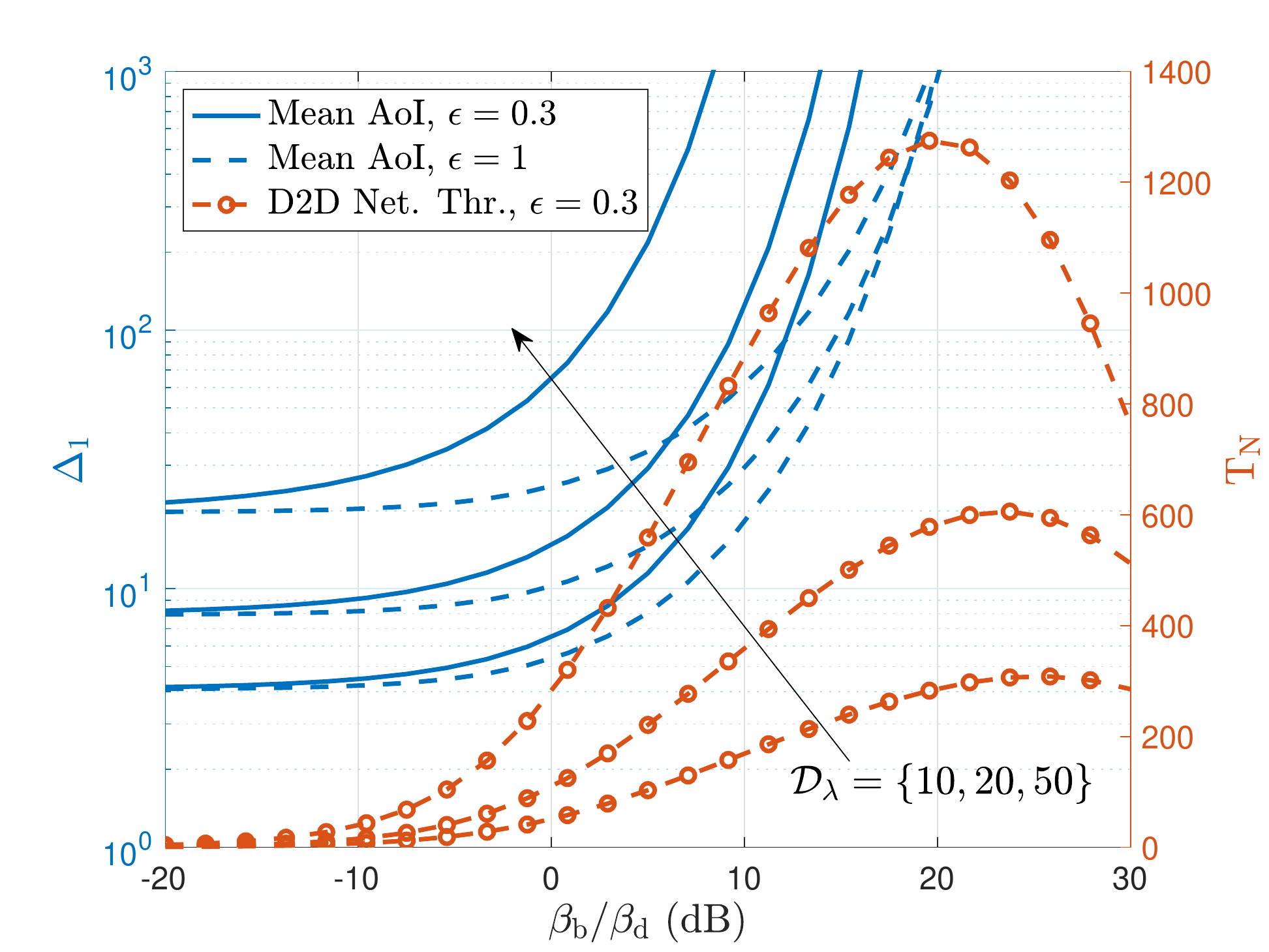}
		\hspace{-2mm}\includegraphics[ width=.33\textwidth]{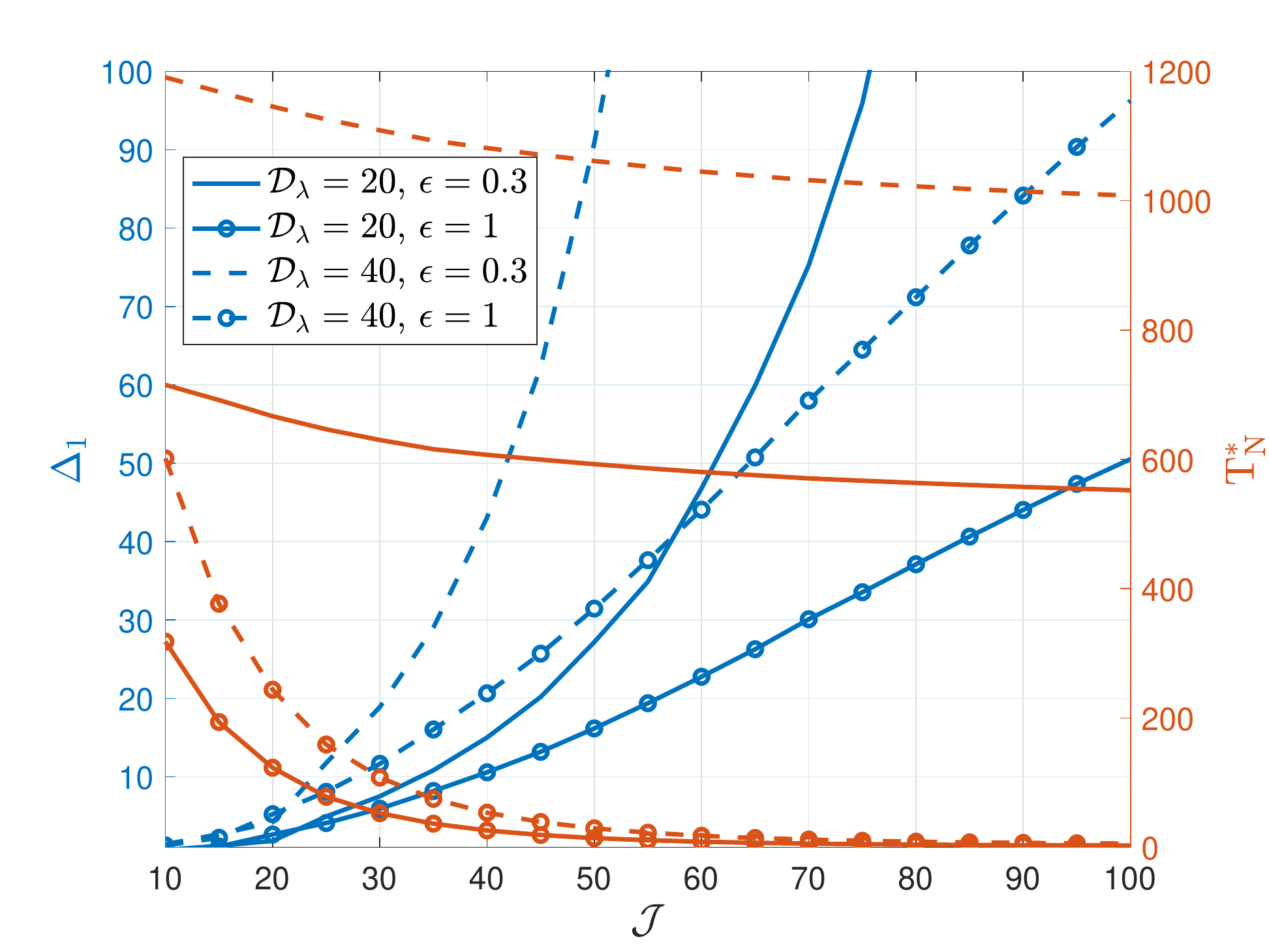}
		\includegraphics[width=.33\textwidth]{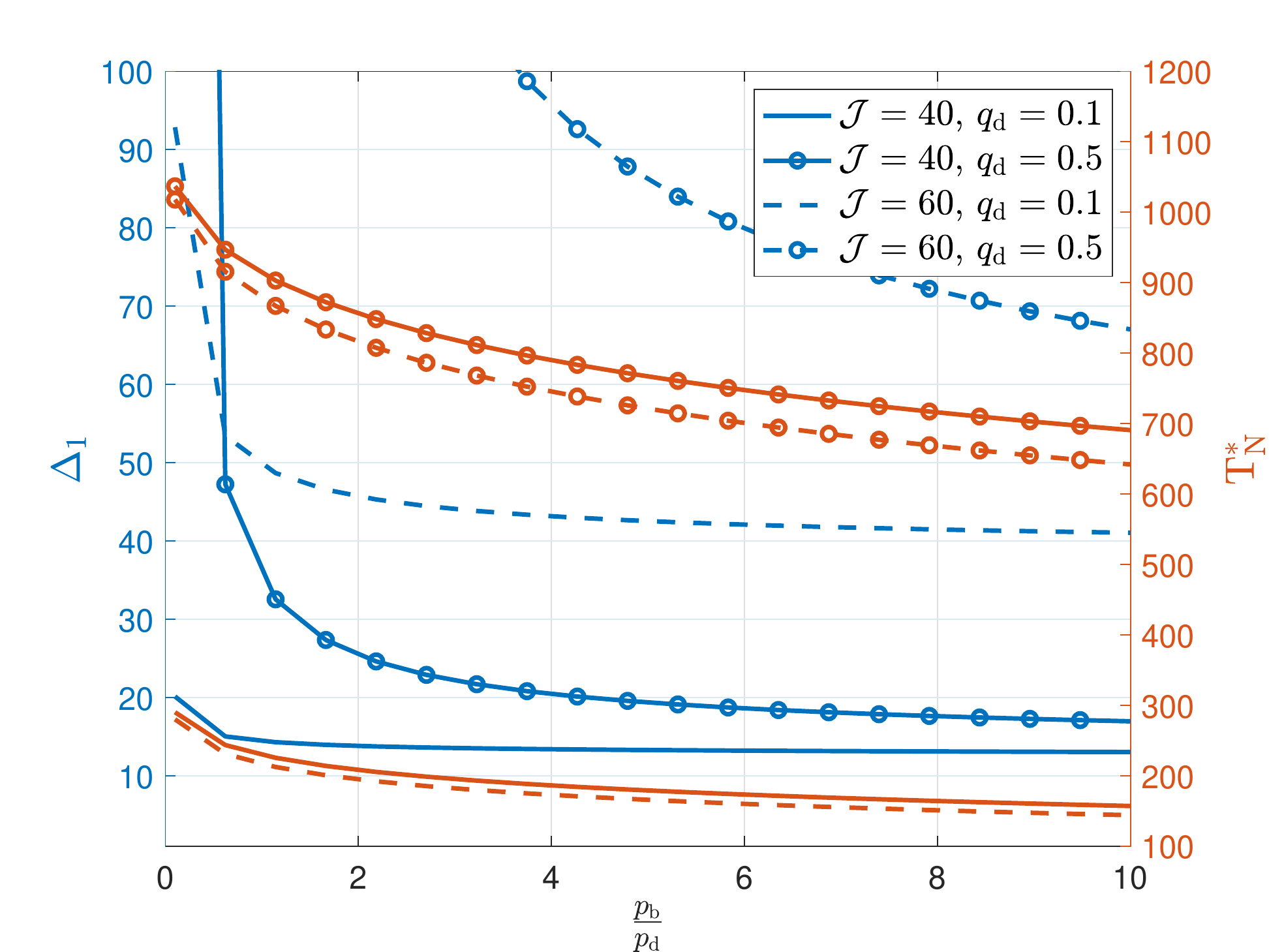}
	\caption{Left: mean AoI and D2D network throughput versus the $\sir$ threshold. Interplay of AoI and D2D transmission rate with respect to JM cell radius $\ncalJ$ (middle) and power ratio $\frac{p_{\rm b}}{p_{\rm d}}$ (right). }
	\label{fig:Interplay_fig}
\end{figure*}

Fig. \ref{fig:Interplay_fig} (left) shows the impact of $\sir$ thresholds on the spatio-temporal mean AoI of status update transmissions and  the throughput of D2D network.
The  mean AoI increases with the $\sir$ threshold $\beta_{\rm b}$, which is expected as the success probability of status updates drops with the increase of $\beta_{\rm b}$.   
The figure shows that the mean AoI  $\Delta_{1}$ is almost equal to $\E[N_{\ncalV_o}]=\mathcal{D}_{\rm \lambda}(1-\exp(-\pi\lambda_{b}\ncalJ^2))$ for a small value of $\beta_{\rm b}$.
{ That is the mean AoI is equal to the mean number of slots required for scheduling the status updates from the typical device when $\beta$ is very small.}
This happens because the success probability of status updates is almost equal to one for small values of $\beta_{\rm b}$ and the mean number of slots required for a device to attempt the transmission is equal to the number of devices  in the associated with the serving BS.  On the other hand, $\Delta_{1}$ rises rapidly as $\beta_{\rm b}$ increases ultimately approaching to a value where the success probability of the status updates  is close to zero (the corresponding points can be confirmed from Fig. \ref{fig:SuccessProb_Pd}), which is expected. However, a finite mean AoI can be supported for large values of $\beta_{\rm b}$ by increasing the power control fraction $\epsilon$. The figure shows the mean AoI curves for the extreme cases of power control (i.e., $\epsilon=0$ and $\epsilon=1$).

    \begin{figure*}[h]
	\centering
	\hspace{-5mm}\includegraphics[ width=.33\textwidth]{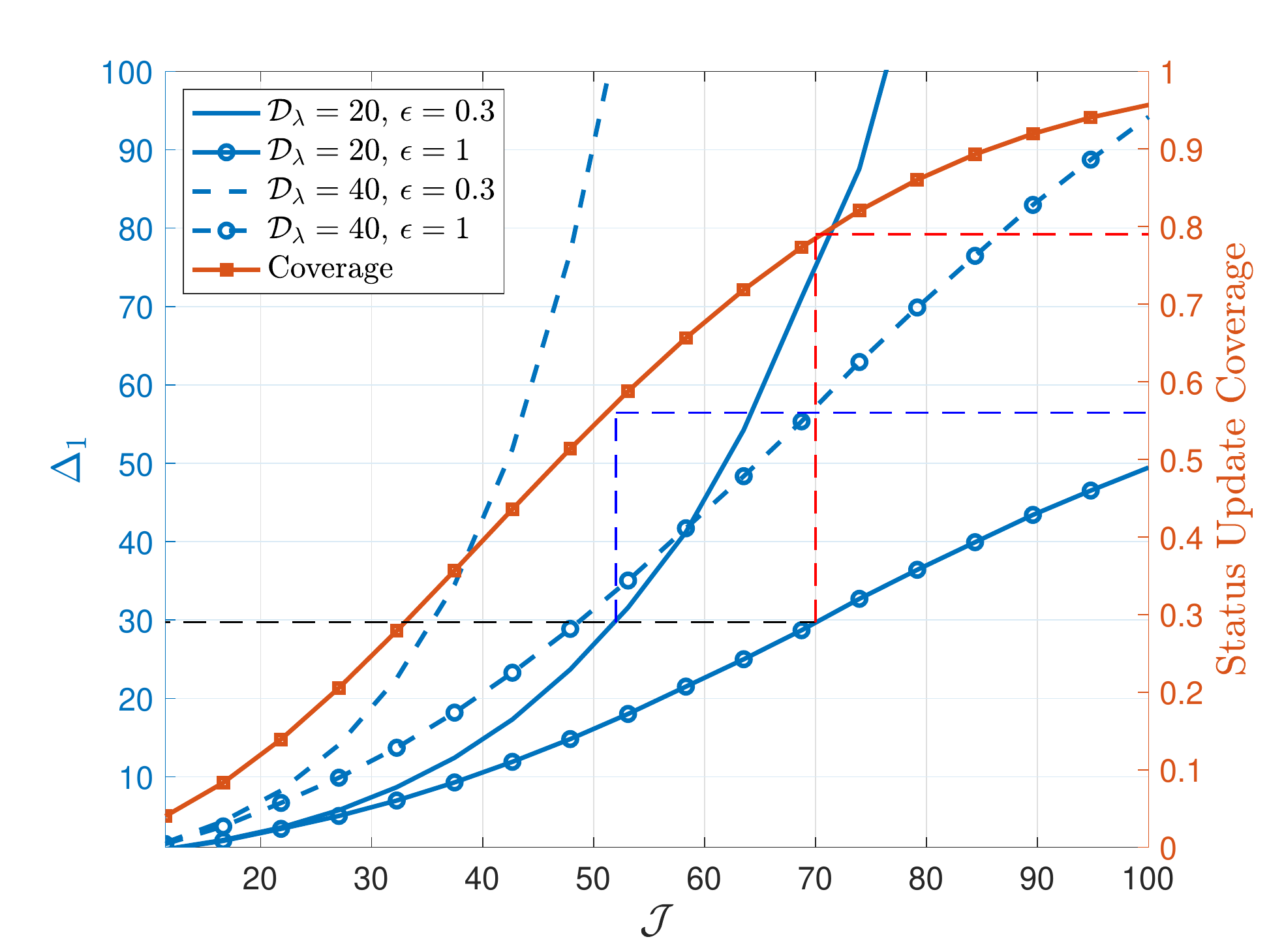}
		\includegraphics[ width=.33\textwidth]{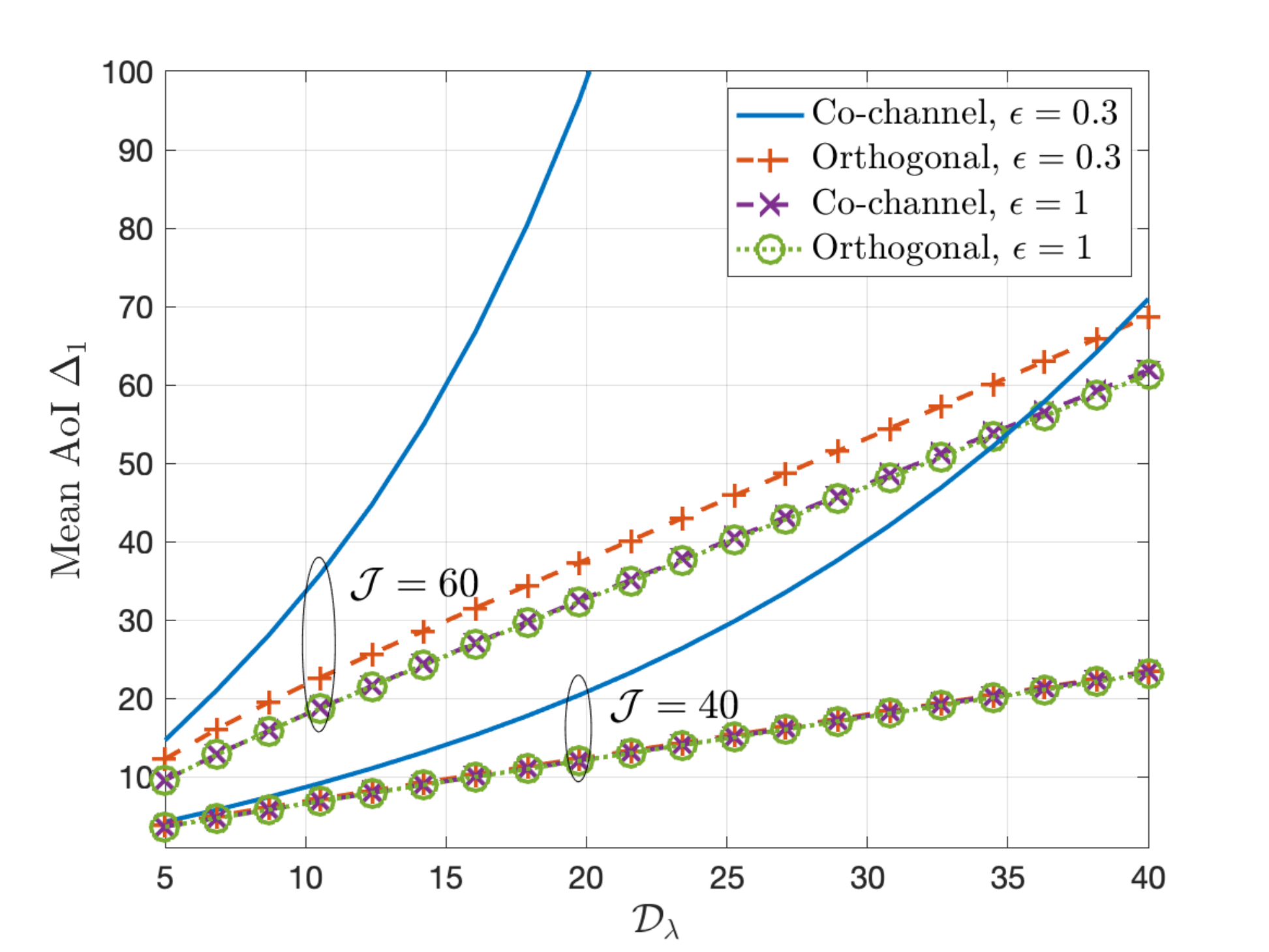}
		\includegraphics[ width=.33\textwidth]{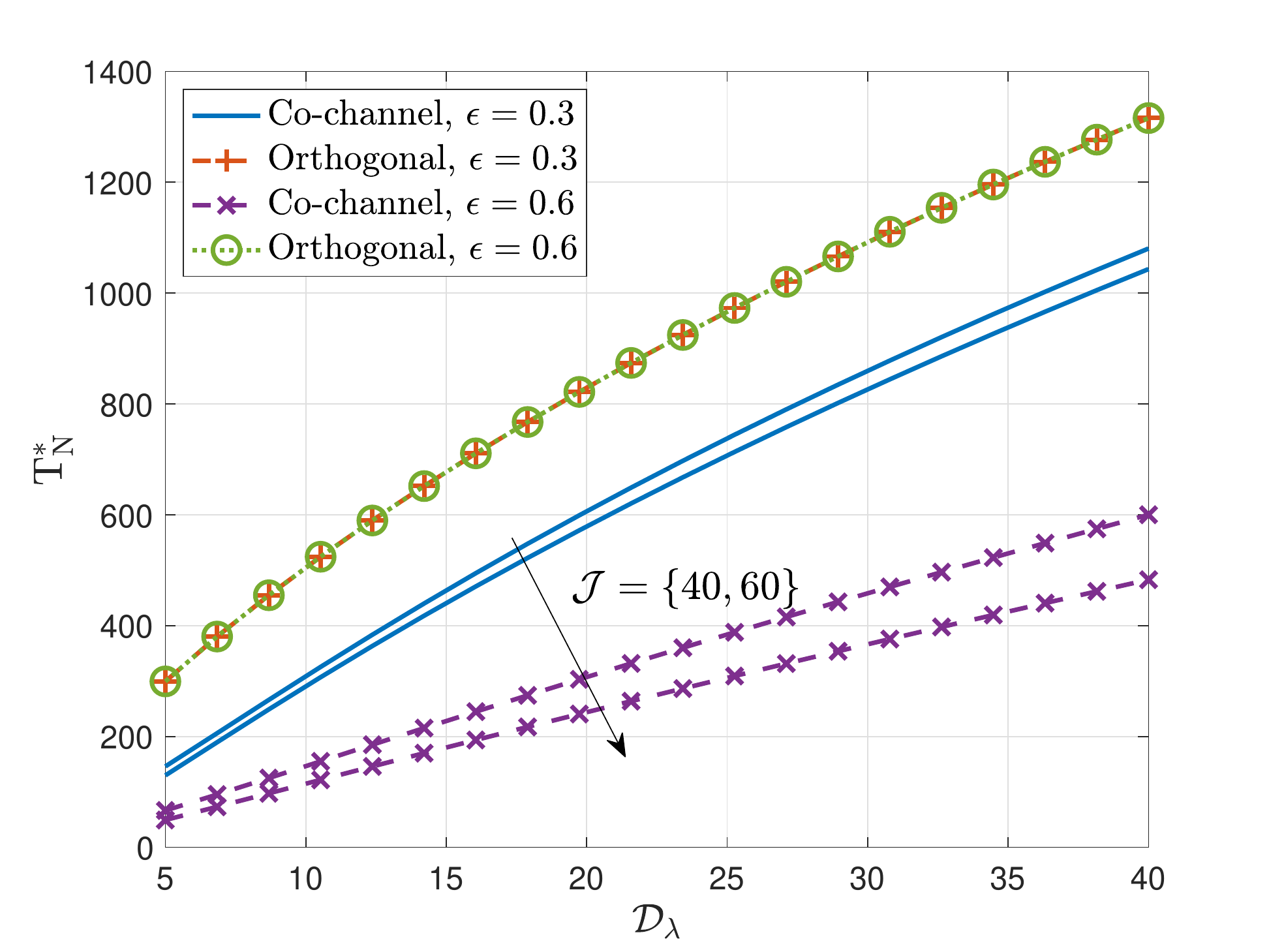}
	    \caption{{ Left: mean AoI versus status update coverage for $\beta_{\rm b}=0$ dB. Middle and right: Co-channel and orthogonal access comparisons for  the mean AoI and the D2D network throughput. }}
	\label{fig:orthogonal_comparision}
    \end{figure*}

The interplay between mean AoI and achievable D2D network throughput with respect to the JM cell radius $\ncalJ$ and the ratio of powers of update and regular transmissions  are presented in Fig. \ref{fig:Interplay_fig} (middle) and  Fig. \ref{fig:Interplay_fig} (right), respectively. Fig. \ref{fig:Interplay_fig} (middle) shows that   both the mean AoI and  D2D throughput degrade with increasing   cell radius $\ncalJ$. 
    With the increase in $\ncalJ$, both the scheduling probability and the success probability drop, which in turn  causes poor AoI performance. In particular, with increasing $\ncalJ$, the scheduling probability decreases because of the need to support status updates for a large number of devices while the success probability drops because of the increase in both the serving link distance and interference. On the other hand, the 
    degradation in the D2D throughput is due to the fact that status updates will need to be transmitted at a higher power because of the increased link distances (with increasing $\ncalJ$), which increases the aggregate interference power.   
    Further, the figure shows that higher $\epsilon$ results in a better AoI performance at the cost of degraded D2D throughput. Therefore, for a given $\lambda_{\rm d}$, we can maximize the D2D network throughput by selecting  minimum  $\epsilon$  that keeps the mean AoI below a predefined performance threshold.     
    Fig. \ref{fig:Interplay_fig} (right) shows that the mean AoI improves and the D2D network throughput degrades with the increase in the ratio $\frac{p_{\rm b}}{p_{\rm d}}$, which is quite expected. However, in this regime, the impact of the increasing power ratio becomes insignificant on the mean AoI since the  interference from the D2D transmission becomes insignificant  (thus  the success probability of status update becomes invariant to $\p_{\rm b}$). 
    It may be noted that both the   D2D throughput and the mean AoI depend on  $p_{\rm b}$ and $p_{\rm d}$ through their ratio.
    

{ For a given $\ncalJ$, the  fraction of devices with status update support (i.e., status update coverage) is equal to $1-\exp(-\pi\lambda_{\rm b}\ncalJ^2)$. 
Fig. \ref{fig:orthogonal_comparision} (left) shows the interplay between the mean AoI and status update coverage.  It particular, it shows that one can tune  $\epsilon$ in the power control model to achieve a higher status update coverage for a given mean AoI target.  For instance, the figure shows that the full power control provides coverage of approximately  80\%, whereas  $\epsilon=0.3$  supports the coverage of approximately  55\%  when the mean AoI threshold is 30 and $\mathcal{D}_{\rm \lambda}=20$. It is worth noting that allowing full status update coverage (i.e., $\ncalJ=\infty$) will result in unbounded mean AoI as the AoI grows rapidly when the conditional success probability approaches to zero. Therefore,  the knowledge of feasible status update coverage is important from the perspective of network design to ensure bounded mean AoI.}


    { Fig.  \ref{fig:orthogonal_comparision} (middle and right) shows that the  mean AoI degrades and the achievable D2D network throughput improves with the increase in $\ncalD_\lambda$, which is expected. 
    From the middle figure, it can be observed that the power control fraction $\epsilon$ does not affect  the mean AoI much under the orthogonal access. Moreover, the co-channel mode with full power control  results in almost equal mean AoI as the orthogonal access case. Thus, the orthogonal access is preferable  when the transmission power is limited, while the co-channel access is preferable when  the spectrum is limited. The right figure shows that the orthogonal access provides higher D2D throughput compared to the co-channel access and the gain increases with $\epsilon$.  }
\section{Conclusion}
This paper presented a stochastic geometry-based analysis of throughput and AoI performance metrics in a cellular-based IoT network while accounting for the spatial disparity in the AoI performance experienced by various wireless links spread across the network. In particular, the throughput was used to characterize the QoS of D2D communications between IoT devices, whereas the AoI was employed to quantify the freshness of status updates (regarding some time-sensitive applications) transmitted by the IoT devices to cellular BSs. The locations of IoT devices and BSs were modeled as a bipolar PPP and an independent PPP, respectively.
Further, we considered that each BS schedules the transmission of status updates from the IoT devices located in its JM cell. In addition, the IoT devices were assumed to employ a distance-proportional fractional power control scheme for uplink transmissions to improve the success delivery rate of status updates. For this setup, the mean success probability for the D2D links was derived to characterize the average network throughput. On the other hand, we captured the spatial disparity in the AoI performance by characterizing spatial moments of the temporal mean AoI. Specifically, we obtained the spatial moments of the temporal mean AoI by deriving the moments of both the conditional success probability and the conditional scheduling probability for status update links.
We validated the analytical results using extensive simulations. Our numerical results demonstrated the impact of power control, medium access probability and density of IoT devices on the achievable D2D network throughput and the spatio-temporal mean AoI. In particular, the results showed that the power control can facilitate the transmission of status updates from a large number of IoT devices such that the mean AoI remains below some predefined threshold.  

{ The analysis of the interplay between AoI and throughput for the case where the IoT devices can employ superposition coding for the non-orthogonal transmission of the regular packets (to other devices) and status updates (to the BSs) could be considered as the direction of this work.}
\appendix
\section*{Proof for Lemma \ref{lemma:Area_Moments}}
\label{appendix:Area_Moments}

Using \cite[Eq. (15)]{robbins1944}, the mean area of $\ncalV_o$ presented in \eqref{eq:mean_Vo} can be directly derived as
\begin{align}
\bar{\ncalV}_o^1&=\int\nolimits_{\ncalB_o(\ncalJ)}\P(\x\in V_o){\rm d}x,\nonumber\\
&=2\pi\int\nolimits_0^\ncalJ \exp(-\pi\lambda_{\rm b}r^2)r{\rm d}r,\nonumber\\
&=\frac{1}{\lambda_{\rm b}}\left(1-\exp\left(-\pi\lambda_{\rm b}\ncalJ^2\right)\right).
\end{align}
 Similarly, using \cite[Eq. (21)]{robbins1944},  the second moment of area of $\ncalV_o$ can be determined as
\begin{align*}
\bar{\ncalV}_o^2&=\int_{\ncalB_{o}(\ncalJ)}\int_{\ncalB_{o}(\ncalJ)}\P({\x,\y}\in\ncalV_o){\rm d}\x{\rm d}\y,\\
&=\int_{\ncalB_{o}(\ncalJ)}\int_{\ncalB_{o}(\ncalJ)}\exp(-\lambda_{\rm b}|\ncalB_\x(\|\x\|)\cap\ncalB_\y(\|\y\|)|){\rm d}\x{\rm d}\y,\\
&=4\pi\int_0^\pi\int_0^\ncalJ\int_0^\ncalJ  \exp(-\lambda_{\rm b}\tilde{S}(r_1,r_2,\psi))r_1r_2 {\rm d}r_1{\rm d}r_2{\rm d}\psi,\\
&=4\pi I,
\end{align*}
where $|A|$ is the area of $A$ and $\tilde{S}(r_1,r_2,\phi)$ is the area of union of two circles centred at $\x$ and $\y$ as shown in Figure \ref{fig:Circle_Union}.  
\begin{figure}[h]
\centering
 \includegraphics[clip, trim=8cm 2cm 10cm 17cm, width=.7\textwidth]{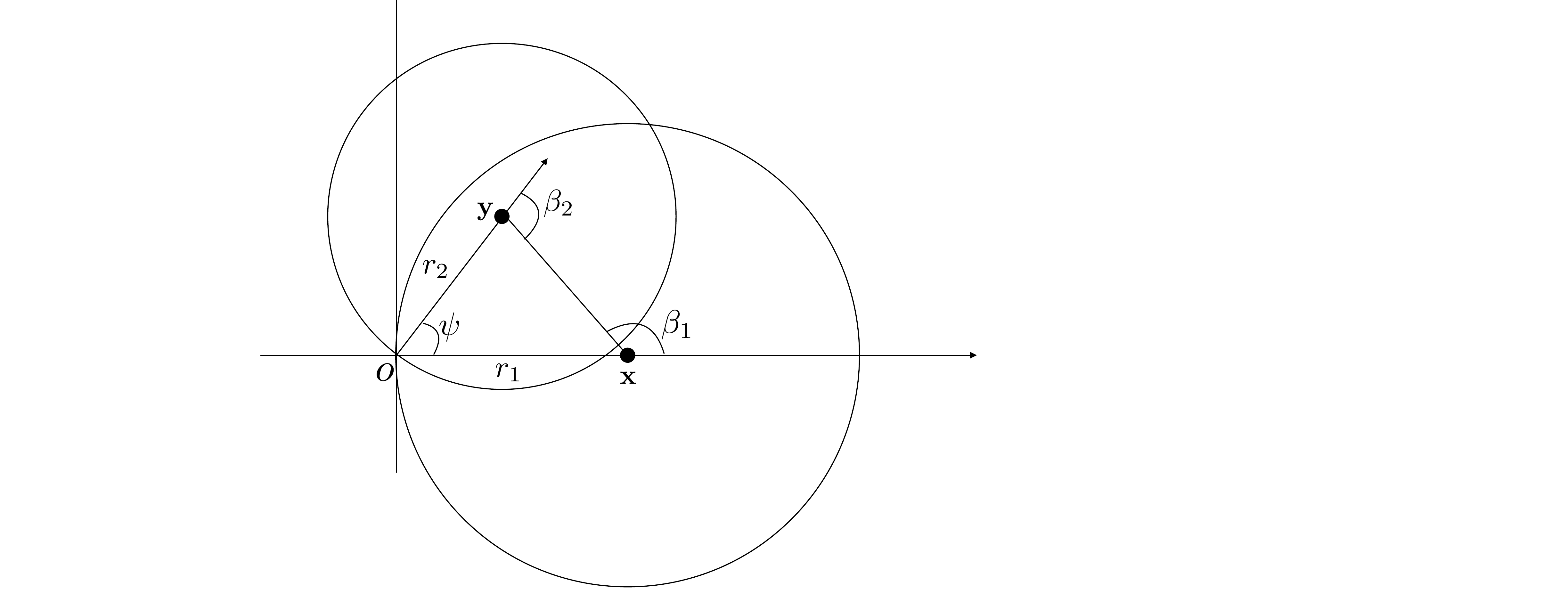}
 \caption{Illustration of union of two circles centred at $\x\equiv(r_1,0)$ and $\y\equiv(r_2,\psi)$.}
 \label{fig:Circle_Union}
\end{figure}
The area of union can be determined as $\tilde{S}(r_1,r_2,\psi)=r_1r_2\sin(\psi) + r_1^2\beta_1 + r_2^2\beta_2$ where $\beta_1$ and $\beta_2$ are the external angles of the triangle ${\rm o\x\y}$ as shown in the figure. 
To evaluate the integral $I$, we employ the change of variables as 
$$r_1=2R\cos(\varphi)\text{~~and~~} r_2=2R\cos(\psi-\varphi).$$
So, we have $\varphi\in\left(-\frac{\pi}{2},\frac{\pi}{2}\right), \psi\in\left(0,\varphi+\frac{\pi}{2}\right), \text{~and~} R\in\left(0,{\rm \tilde{R}_{\max}}\right)$ where ${\rm \tilde{R}_{\max}}=0.5\ncalJ\max(\cos(\varphi),\allowbreak \cos(\psi-\varphi))^{-1}$ and the determinant of the Jackobian Matrix can be determined as
\begin{align*}
\bigg|\frac{D(r_1,r_2)}{D(R,\varphi)}\bigg|=4R\sin(\varphi).
\end{align*}
With this change of variables, the external angles of triangle ${\rm o\x\y}$ can be expressed as
$\beta_1=\frac{\pi}{2}-\varphi+\psi\text{~and~}\beta_2=\frac{\pi}{2}+\varphi$. Therefore, we get
\begin{align*}
I=16\int_{-\frac{\pi}{2}}^{\frac{\pi}{2}}\int_{0}^{\varphi+\frac{\pi}{2}}& \int_0^{{\rm \tilde{R}_{\max}}} R^3\exp\left(-4R^2\hat{S}(\varphi,\psi)\right)\\
&\times\cos(\varphi)\cos(\psi-\varphi)\sin(\varphi) {\rm d}R{\rm d}\psi {\rm d}\varphi
\end{align*}
where $\hat{S}(\varphi,\psi)=\cos(\varphi)\cos(\psi-\varphi)\sin(\psi) +\left(\frac{\pi}{2}-\varphi+\psi\right) \cos(\varphi)+\left(\frac{\pi}{2}+\varphi\right)\cos(\psi-\varphi).$
Solving the inner most integral of $I$ w.r.t $R$, we get
 \begin{align*}
I=\frac{1}{2\lambda_{\rm b}^2}\int_{-\frac{\pi}{2}}^{\frac{\pi}{2}}&\int_{0}^{\varphi+\frac{\pi}{2}} \frac{\hat{G}(\varphi,\psi)}{\hat{S}(\varphi,\psi)^2}\bigg[1-\left(1+\lambda_{\rm b}\ncalJ^2S^{\prime\prime}(\varphi,\psi)\right)\\
&\times\exp\left(-\lambda_{\rm b}\ncalJ^2S^{\prime\prime}(\varphi,\psi)\right)\bigg]{\rm d}\psi {\rm d}\varphi,
\end{align*}
where $\hat{G}(\varphi,\psi)=\cos(\varphi)\cos(\psi-\varphi)\sin(\psi)$ and $S^{\prime\prime}(\varphi,\psi)=\frac{\hat{S}(\varphi,\psi)}{\max(\cos(\varphi),\cos(\psi-\varphi))^2}$.
Finally, by substituting $\varphi=\frac{\pi}{2}-u$ and $\psi-\varphi=\frac{\pi}{2}-v$ and further simplifying, we obtain the second moment of area of $\ncalV_o$ as given in \eqref{eq:2nd_moment_Vo}.

\end{document}